\documentclass[aps,prx,10pt,twocolumn,floatfix]{revtex4-2}
\usepackage[utf8]{inputenc}
\usepackage{graphicx}
\usepackage{amsmath}
\usepackage{amssymb}
\usepackage{tikz}
\usepackage{enumitem}
\usepackage{bbm}
\usepackage{mathtools}
\usepackage{units}
\usepackage[ruled, linesnumbered]{algorithm2e}
\usepackage{amsthm}
\usepackage{xcolor}
\usepackage{hyperref}
\usepackage[capitalise]{cleveref}
\hypersetup{
    colorlinks=true,
    linkcolor=cyan,           % color of internal links
    citecolor=cyan,           % color of links to bibliography
    urlcolor=cyan
}

\newtheorem{theorem}{Theorem}[section]
\newtheorem{example}{Example}[section]
\newtheorem{corollary}{Corollary}[section]
\newtheorem{lemma}[theorem]{Lemma}
\newtheorem*{lemma2}{Lemma}

\newtheorem{remark}{Remark}[section]

\newcommand{\X}{\mathcal{X}}
\newcommand{\config}{\boldsymbol{\sigma}}
\newcommand{\calH}{\mathcal{H}}
\newcommand{\XX}{\mathcal{X}\times\mathcal{X}}
\newcommand{\C}{\mathcal{C}}
\begin{document}

\title{The Houdayer Algorithm: Overview, Extensions, and Applications}

\author{Adrien Vandenbroucque}
\email{adrien@entropicalabs.com}
\author{Ezequiel Ignacio Rodr\'iguez Chiacchio}
\email{ezequiel@entropicalabs.com}
\author{Ewan Munro}
\email{ewan@entropicalabs.com}
\affiliation{Entropica Labs, 186b Telok Ayer Street, Singapore 068632}

\begin{abstract}
The study of spin systems with disorder and frustration is known to be a computationally hard task. Standard heuristics developed for optimizing and sampling from general Ising Hamiltonians tend to produce correlated solutions due to their locality, resulting in a suboptimal exploration of the search space. To mitigate these effects, cluster Monte-Carlo methods are often employed as they provide ways to perform non-local transformations on the system. In this work, we investigate the Houdayer algorithm, a cluster Monte-Carlo method with small numerical overhead which improves the exploration of configurations by preserving the energy of the system. We propose a generalization capable of reaching exponentially many configurations at the same energy, while offering a high level of adaptability to ensure that no biased choice is made. We discuss its applicability in various contexts, including Markov chain Monte-Carlo sampling and as part of a genetic algorithm. The performance of our generalization in these settings is illustrated by sampling for the Ising model across different graph connectivities and by solving instances of well-known binary optimization problems. We expect our results to be of theoretical and practical relevance in the study of spin glasses but also more broadly in discrete optimization, where a multitude of problems follow the structure of Ising spin systems.
\end{abstract}

\maketitle

\section{Introduction}
    Discrete optimization problems arise extensively across numerous scientific disciplines and industrial applications~\cite{operations_research_intro, operations_research_and_ai}. An attractive feature of these problems is that they can often be expressed as quadratic unconstrained binary optimization (QUBO) problems, or equivalently, Ising Hamiltonians~\cite{ising_formulation_np_problems, qubo_biology, qubo_chemistry, qubo_logistics, qubo_machine_learning}. Despite the simplicity of their formulation, these problems are NP-hard in general and there is thus no known polynomial-time exact algorithm to solve them~\cite{spin_glasses_complexity}. 

A plethora of heuristics have been developed to gain deeper understanding of these systems, ranging from Monte-Carlo methods~\cite{intro_to_mcmc} to genetic algorithms~\cite{genetic_algo_ising}. Most common approaches, like the Metropolis algorithm~\cite{metropolis_paper}, require few computational resources because they only perform local changes on the state of the system. However, these tend to produce correlated solutions, and can therefore act as biased samplers~\cite{ground_state_stats_annealing_algorithms, exp_biased_ground_state_sampling, no_fair_sampling_quantum_annealing}. This motivates the development of cluster Monte-Carlo algorithms, which perform non-local transformations by altering the state of multiple spins at a time. Leveraging this strength, methods like the Swendsen-Wang~\cite{WANG1990565} and Wolff algorithms~\cite{wolff} are able to mitigate the impact of long autocorrelation times~\cite{critical_slowing_down}, thus allowing to reduce the computational effort needed to simulate spin systems~\cite{beyond_suppression_critical_slowing_down}. Nevertheless, these approaches become less effective in the presence of disorder and frustration, especially when considering high-dimensional lattices~\cite{cluster_algo_spin_glasses}.

The Houdayer algorithm is a cluster Monte-Carlo method that has recently gained strong interest due to its capacity to navigate complex energy landscapes~\cite{Houdayer_2001,cluster_algo_any_space_dimension,icm_generating_more_optimum,quantum_assisted_genetic_algorithm,multiqubit_correction,icm_fair_sampling}. This is achieved by performing non-local moves which preserve the energy of the system and enable ``tunnelling" through large energy barriers. As a result, it has become one of the primary algorithms to solve generic Ising and QUBO problems~\cite{borealis}. Multiple methods exploiting it have also been developed to improve the search of low-energy states~\cite{icm_generating_more_optimum, quantum_assisted_genetic_algorithm, multiqubit_correction} and to obtain a fairer sampling of ground-state configurations~\cite{icm_fair_sampling}. Additionally, for spin-glass Hamiltonians, the Houdayer algorithm can lead to faster thermalization as compared to conventional methods~\cite{Houdayer_2001,cluster_algo_any_space_dimension}. Given these capabilities, it has been used as a reference point in a variety of works benchmarking quantum annealing and population annealing algorithms~\cite{fujitsu_annealer,spin_glass_good_bad_ugly, population_annealing_icm}, as well as quantum-enhanced Monte-Carlo strategies~\cite{ibm_paper}. Further studies have also proposed several variations of it, allowing its use in any space dimension~\cite{cluster_algo_any_space_dimension} or when considering diluted lattices~\cite{cluster_move_diluted_spins, cluster_move_diluted_spins_critical_behavior}. While offering enticing features, the Houdayer algorithm in its original form suffers from a number of impediments, such as a bias towards selecting large clusters and difficulties operating on certain problem geometries~\cite{percolation_signature_spin_glass}.

In this work, we investigate the mechanisms of the Houdayer move and develop tools to overcome its limitations. Our analysis naturally leads to a generalization of the procedure, which we leverage in order to design methods that offer a better exploration of the space of configurations. Some applications of these techniques include sampling from Boltzmann distributions or approximating the ground-state manifold of Ising systems. We also focus on more practical settings such as industrial optimization use-cases, where the main target is often to find a set of reasonable diverse solutions rather than just the optimal one~\cite{optimization_diverse_solutions}. For this reason, we build upon previous work~\cite{icm_generating_more_optimum, multiqubit_correction} and introduce a genetic algorithm which uses Houdayer moves at its core to mix configurations together. This is particularly relevant in a number of scenarios like generating optimal configurations given an initial set of them, or producing feasible solutions in constrained optimization problems.

The rest of the paper is organized as follows. \Cref{section:background} introduces the model considered in this work and the conventions used throughout the text. \Cref{section:original_houdayer_review} is dedicated to a review of the original Houdayer algorithm, highlighting its properties, advantages, and pitfalls. We introduce our generalization in \cref{section:generalized_houdayer} and compare its properties with those of the original algorithm. In \cref{section:applications}, a number of applications are presented, including the use of the algorithm in Markov chain Monte-Carlo methods and the development of a genetic algorithm to improve optimization and sampling. These schemes are put to test in \cref{section:numerical_results}, where we report the performance obtained when sampling from various Ising systems and solving well-known optimization problems.

\section{Background and Definitions} \label{section:background}
    \subsection{Model}
Let $G=(V, E)$ be a graph, and let $\mathcal{X} = \{-1, 1\}^{|V|}$ be the space of all possible assignments of the values $\{-1, 1\}$ to the vertices. In this work, we consider the problem of minimizing a function $\mathcal{H}:\mathcal{X} \to \mathbb{R}$ of the form
\begin{equation}\label{equation:problem_statement}
    \mathcal{H}(\boldsymbol{\sigma}) = -\sum_{(i, j) \in E} J_{ij} \sigma_i \sigma_j - \sum_{i \in V} h_i \sigma_i,
\end{equation} 
where $\boldsymbol{\sigma}\in \mathcal{X}$ and $J_{ij}, h_i\in \mathbb{R}$. 
\iffalse
This is equivalent to writing $\mathcal{H}(\boldsymbol{\sigma}) = \boldsymbol{\sigma}^\top J \boldsymbol{\sigma} + \boldsymbol{h}^\top\boldsymbol{\sigma}$, where $J\in \mathbb{R}^{|V| \times |V|}$ is a strictly upper-triangular matrix containing the values $J_{ij}$ for each $(i, j)\in E$, and the vector $\boldsymbol{h}\in\mathbb{R}^{|V|}$ collects the values $h_1, \dots, h_{|V|}$. 
\fi
We denote by $\mathcal{H}_0=\min_{\boldsymbol{\sigma}\in\mathcal{X}}\mathcal{H}(\boldsymbol{\sigma})$ the minimum value the function can take. Since such a function may have multiple minima, we denote by $\mathcal{X}_0 = \{\boldsymbol{\sigma} \in \mathcal{X}: \mathcal{H}(\boldsymbol{\sigma}) = \mathcal{H}_0\}$ the subset of $\mathcal{X}$ whose elements attain the value $\mathcal{H}_0$. In the following, any reference to some graph $G$ assumes a function $\mathcal{H}$ of interest from which the graph stems. Unless stated otherwise, any mention of some function $\mathcal{H}$ will assume that is in of the form of \cref{equation:problem_statement}.

This well-known class of functions corresponds to Ising Hamiltonians and is used to assign energy to spin systems in the Ising model~\cite{ising_model}. Much of our notation is borrowed from the field of statistical physics and we sometimes refer to $\mathcal{H}$ as Ising Hamiltonian or Ising formulation, but we hope to remain more general in order to emphasize the use of Houdayer moves in optimization procedures that deviate from the more traditional Monte-Carlo methods encountered in the literature.
\begin{remark} \label{remark:ising_one_to_one_mapping}
    More generally, any function of the form of \cref{equation:problem_statement} taking input values $\boldsymbol{\sigma} \in \{a, b\}$ with $a, b \in \mathbb{R}$ and $a \neq b$ can be converted into an equivalent Ising formulation through the one-to-one mapping $f:\{a, b\} \to \{-1, 1\}$ given by the affine function $f(x) = \frac{2}{b-a}x - (b+a)$. This is for example useful when converting QUBO problems into their Ising equivalent.
\end{remark}
We now summarize the nomenclature used throughout the text. For the growth rate of running times, we adopt the standard asymptotic notation $O, \Theta$, and $\Omega$~\cite[Section 2.2]{algorithms_design}. \iffalse For the growth rate of running times, we adopt the notation in~\cite[Section 2.2]{algorithms_design}: given a function $f$ to be estimated (representing the running time of some computation) and a comparison function $g$, we say that $f(n)$ is $O(g(n))$ if there exist constants $c>0$ and $n_0\geq 0$ such that $f(n)\leq c\cdot g(n)$ for all $n\geq n_0$. We say $f(n)$ is $\Omega(g(n))$ if there exist constants $\epsilon>0$ and $n_0\geq 0$ such that $f(n)\geq \epsilon \cdot g(n)$ for all $n\geq n_0$. Finally, we say $f(n)$ is $\Theta(g(n))$ if $f$ is both $O(g(n))$ and $\Omega(g(n))$.\fi For an arbitrary finite set $S$, we denote its power set (i.e., the set of all subsets of $S$) by $2^S$ and the set of probability distributions over it by $\mathcal{P}(S)$. The Kronecker delta function with parameters $i, j\in \mathbb{R}$ is denoted as $\delta_{ij}$, and the indicator function of some set $A$ is denoted as $\mathbbm{1}_{A}$, so that $\mathbbm{1}_{A}(x)=\begin{cases}1,\quad\text{if } x\in A\\0,\quad\text{otherwise}\end{cases}$. When considering the space $\X$, the notion of distance used is the Hamming distance, defined between two configurations $\boldsymbol{\sigma}, \boldsymbol{\tau} \in \X$ as 
\begin{equation}\label{equation:hamming_distance}
    d_{H}(\boldsymbol{\sigma}, \boldsymbol{\tau}) = \sum_{i=1}^n (1-\delta_{\sigma_i \tau_i}).
\end{equation}
The Hamming distance $d_H(\config, T)$ between a configuration $\config \in \X$ and a set of configurations $T\subseteq \X$ is defined as
\begin{equation}
    d_H(\config, T) = \min_{\boldsymbol{\tau} \in T} d_H(\config, \boldsymbol{\tau}).
\end{equation}
Given a finite set of numbers $P\subseteq \mathbb{R}$, the median of the set is denoted as $\mathrm{med}(P)$.
Given a probability measure $\mu$ over some finite set $S$, we denote the expected value of a function $f: S\to \mathbb{R}$ with respect to $\mu$ as 
\begin{equation*}
    \mathbb{E}_{\mu}[f] = \sum_{s\in S}\mu(s)\cdot f(s).
\end{equation*}

\section{Houdayer Algorithm: Review}\label{section:original_houdayer_review}
    We start with a review of the Houdayer algorithm as introduced in \cite{Houdayer_2001}, with the aim of highlighting its properties and limitations.

\subsection{Description and Properties} \label{section:description_houdayer}
Considering the setup presented in \cref{section:background}, the move is executed through the following steps:
\begin{enumerate}
    \item Given a pair of independent configurations of variables $(\boldsymbol{\sigma}^{(1)}, \boldsymbol{\sigma}^{(2)})\in \mathcal{X}\times\mathcal{X}$, compute the local overlap $q_i = \sigma_i^{(1)}\sigma_i^{(2)}$ at every site $i\in V$, resulting in two domains where the overlap value is either -1 or +1.
     
     \item The overlap computed in the previous step defines clusters, which are the connected parts having the same overlap value (we say that two sites $i,j\in V$ are connected if $(i, j)\in E$).
     Select at random a site for which $q_i=-1$ and, in both configurations, flip the values of all the spins in the cluster to which it belongs, yielding a new pair of configurations $(\boldsymbol{\tau}^{(1)}, \boldsymbol{\tau}^{(2)})\in \mathcal{X}\times\mathcal{X}$.
\end{enumerate}
We will refer to clusters with associated value $q_i=-1$ in the local overlap as \textit{Houdayer clusters} (see \cref{remark:houdayer_quadratic_only}). Overall the Houdayer algorithm, as presented originally, has the purpose of providing a random transition from a given pair of configurations to another one, as depicted in the left side of \cref{fig:houdayer}.
\begin{figure*}[ht]
    \begin{minipage}{\columnwidth}
    \includegraphics[width=0.95\columnwidth]{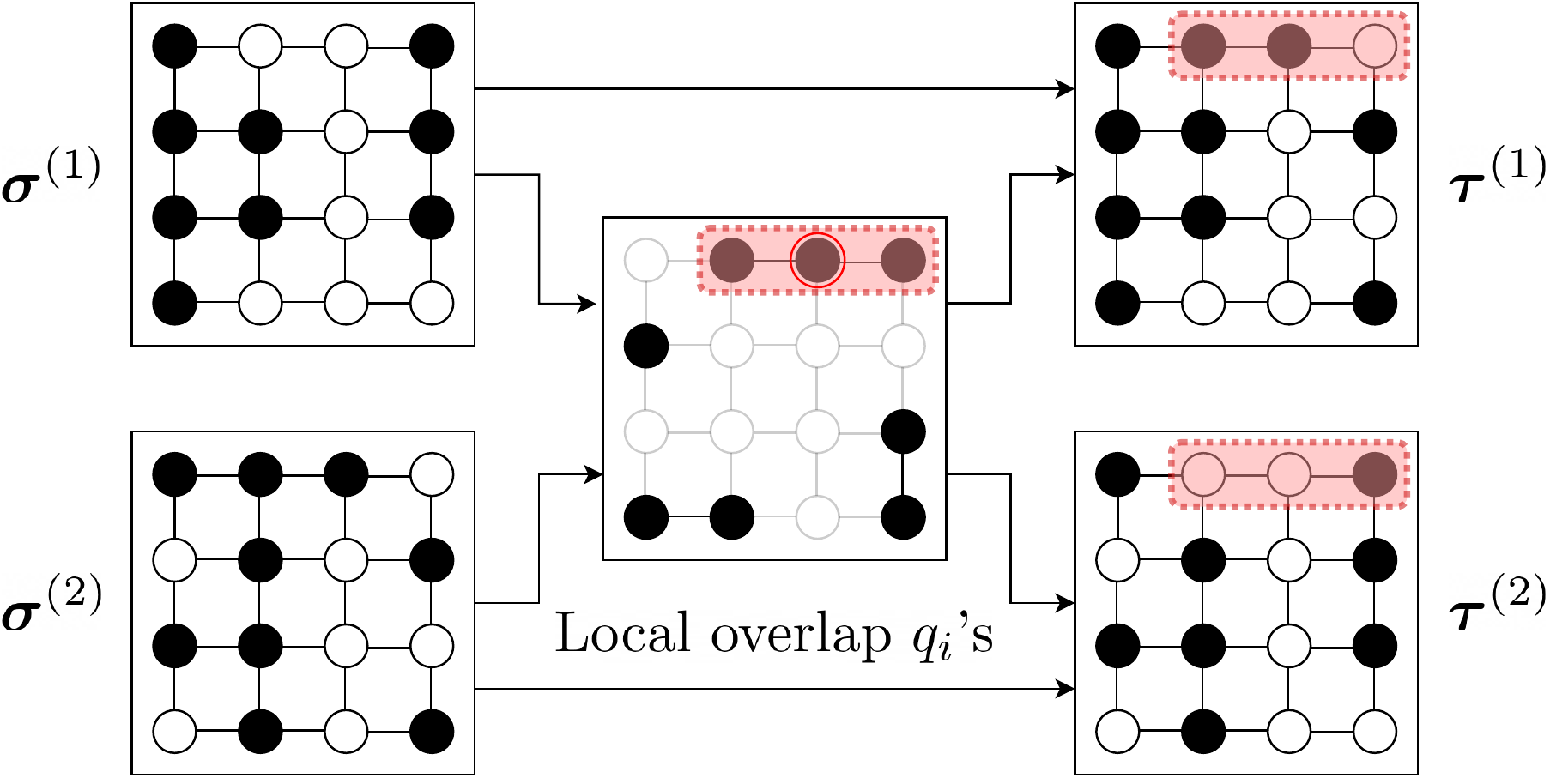}
    \put(-170, 130){\textbf{(a)} Original Houdayer Move}
    \end{minipage}
    \hspace{0.02\columnwidth}
    \unskip \vrule
    \hspace{0.02\columnwidth}
    \begin{minipage}{\columnwidth}
    \includegraphics[width=0.95\columnwidth]{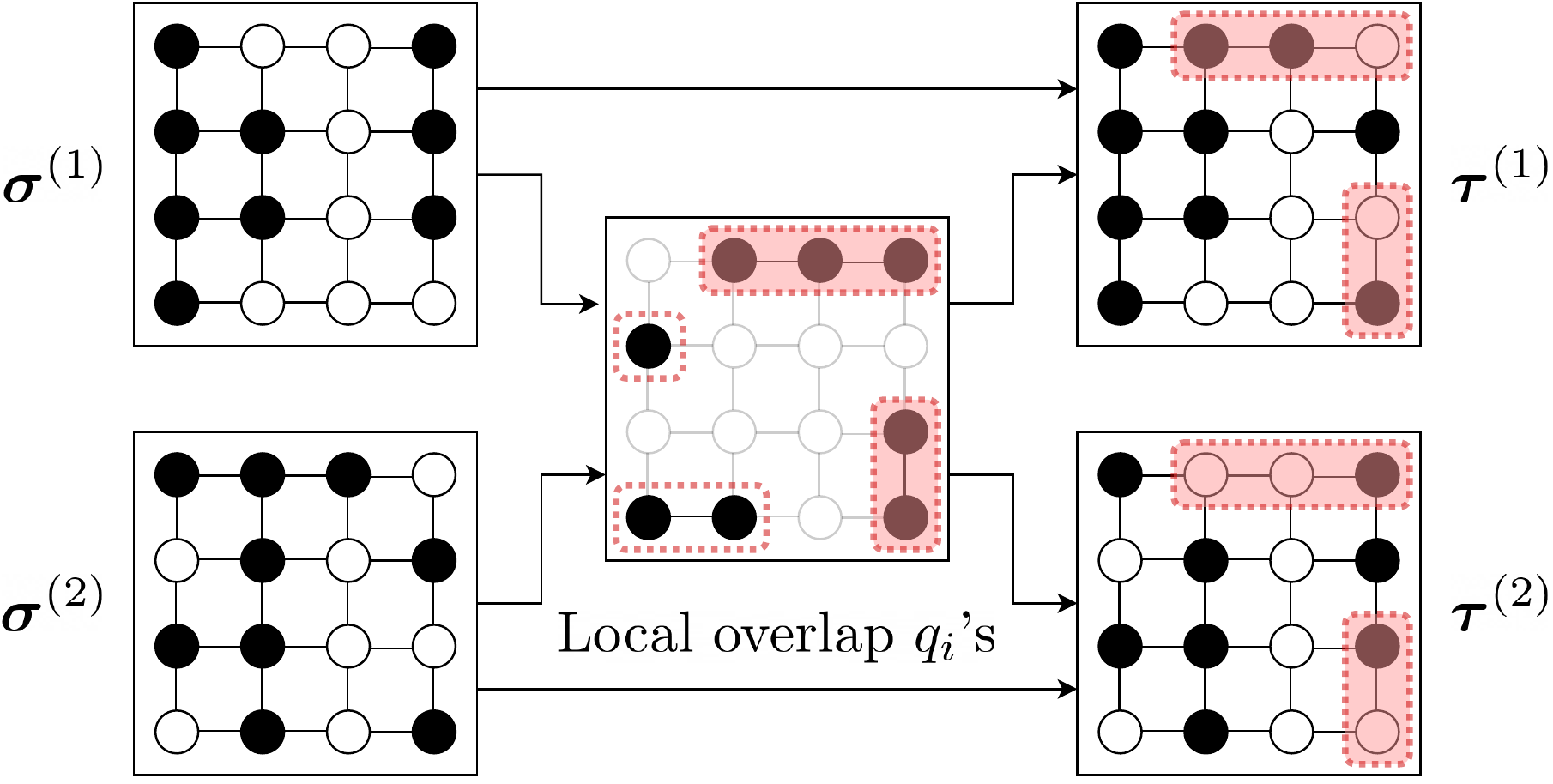}
    \put(-170, 130){\textbf{(b)} Generalized Houdayer Move}
    \end{minipage}
    \caption[Houdayer Moves]{Visual representations of two different versions of the Houdayer move on a $4\times 4$ square lattice graph. Black nodes correspond to variables with value -1, and white nodes to variables with value 1. Both start with a pair of independent configurations of variables $(\boldsymbol{\sigma}^{(1)}, \boldsymbol{\sigma}^{(2)})$ and first compute the local overlap, revealing where variables differ between the two configurations, and defining a set of Houdayer clusters (highlighted with red dashed lines). (a) For the original Houdayer move, one of the sites with overlap value $q_i=-1$ is chosen uniformly at random (circled in red) and the corresponding cluster is grown from there (highlighted in red). (b) For the generalized Houdayer move, a combination of Houdayer clusters is chosen (highlighted in red) at random according to a desired distribution $\nu_{\C}$. Finally, the variables whose sites belong to the selected cluster(s) are flipped, yielding a new pair $(\boldsymbol{\tau}^{(1)}, \boldsymbol{\tau}^{(2)})$.}\label{fig:houdayer}
  \end{figure*}
\begin{remark}
To handle the case of variables taking values in an arbitrary binary alphabet $\{a, b\}$ with $a, b \in \mathbb{R}$ and $a\neq b$, one can more generally define the local overlap at site $i \in V$ as $q_i = \begin{cases}
       1 &\quad\text{if } \sigma_i^{(1)} = \sigma_i^{(2)},\\
       -1 &\quad\text{if } \sigma_i^{(1)} \neq \sigma_i^{(2)}. \\ 
     \end{cases}$. With this definition, step 2 above becomes equivalent to exchanging variables values between the two configurations at sites belonging to the selected Houdayer cluster.
\end{remark}
\begin{remark}\label{remark:houdayer_quadratic_only}
For Ising Hamiltonians with no linear terms (i.e., $h_i=0$), one can also choose in the local overlap clusters with overlap value $q_i=1$ without altering the property that the total function value over the configurations remains constant~\cite{Houdayer_2001}. This however only holds when considering problems in which variables take binary values in $\{-a, a\}$ for some $a \in \mathbb{R}$. Indeed in this case the one-to-one mapping mentioned in \cref{remark:ising_one_to_one_mapping} becomes linear, which does not create linear terms when converting to an Ising formulation.
\end{remark}
The crucial property of this transformation resides in its ability to preserve the total function value of the pair of configurations given as input, as shown in the corollary below.
\begin{corollary}
Let $\mathcal{H}:\mathcal{X}\to \mathbb{R}$ be a function of the form of \cref{equation:problem_statement}. Then for any pair of configurations of variables $(\boldsymbol{\sigma}^{(1)}, \boldsymbol{\sigma}^{(2)})$, the pair of configurations $(\boldsymbol{\tau}^{(1)}, \boldsymbol{\tau}^{(2)})$ resulting from a Houdayer move is such that $\mathcal{H}(\boldsymbol{\sigma}^{(1)}) + \mathcal{H}(\boldsymbol{\sigma}^{(2)}) = \mathcal{H}(\boldsymbol{\tau}^{(1)}) + \mathcal{H}(\boldsymbol{\tau}^{(2)})$.
\end{corollary}
\begin{proof}
The proof follows from the proof of \cref{lemma:extended_houdayer_same_energy} (see \cref{appendix:proof_extended_houdayer_properties}) since the original Houdayer move is a special case of our generalization, for which the same property holds.
\end{proof}
This feature, together with being able to perform large rearrangements of variables, is precisely what makes the Houdayer move an attractive choice in sampling procedures~\cite{Houdayer_2001, cluster_algo_any_space_dimension}. In particular when sampling via Markov chain Monte-Carlo methods, a Houdayer move is always accepted and thus requires small numerical overhead since no random number needs to be generated. Below, we describe two simple properties of the move which are informative of its performance.
\begin{lemma}
Consider a pair of configurations of variables $\left(\boldsymbol{\sigma}^{(1)}, \boldsymbol{\sigma}^{(2)}\right)$ such that their local overlap induces the set of Houdayer clusters $\mathcal{C}$. Then, when performing a Houdayer move, we have that:
\begin{enumerate} \label{lemma:simple_houdayer_properties}
    \item The move is able to reach $|\mathcal{C}| + 1$ unique pairs (including the starting one).
    
    \item \label{property:houdayer_reverse}If $|\mathcal{C}| = 1$, performing the move leads to the reversed pair $\left(\boldsymbol{\sigma}^{(2)}, \boldsymbol{\sigma}^{(1)}\right)$.
\end{enumerate}
\end{lemma}
\begin{proof}[Proof Sketch]
For the first part, since there are $|\mathcal{C}|$ Houdayer clusters to choose from and each of these cluster covers sites where configurations differ, each choice yields a different pair of configurations. In case there is no Houdayer cluster to be chosen in the local overlap (which happens when the two elements of the pair are identical), the resulting pair is the same as the initial one. This amounts to a total of $|\C|+1$ unique pairs.

For the second part, note that if there is only one cluster to choose, then this means that swapping the corresponding variables between the two configurations is equivalent to just exchanging them altogether since a Houdayer cluster represent where variables do not agree between the two configurations.
\end{proof}
The previous lemma gives a precise value for the number of pairs that can be accessed via a Houdayer move. One immediate question that arises is whether this number can be increased, which is the focus of subsequent sections. The second point of \cref{lemma:simple_houdayer_properties} may be seen as some form of limitation since the two configurations arising from the move are the same as the input ones. We discuss this aspect in more depth in the next section. 

\subsection{Limitations}\label{section:houdayer_limitations}
There are a number of important limitations in the original Houdayer move described above.
\begin{itemize}
    \item \textit{The Houdayer move can fail}: Whenever the local overlap induces only one Houdayer cluster, we know from \cref{lemma:simple_houdayer_properties} that the move will simply yield the same pair but reversed, essentially doing nothing. This already raises two questions: how often does this happen, and can it be avoided? A partial answer to the first question resides in how connected $G$ is. Indeed, in the limit where we consider a complete graph, regardless of the local overlap, all sites with local overlap value -1 are necessarily connected, thus forming a single Houdayer cluster. More generally, consider a uniformly random pair of configurations. The clusters formed in the local overlap can be seen as the result of a site percolation process on $G$ with occupation probability $\frac12$ (cf. \cref{appendix:percolation_theory}). Thus, for any graph with percolation threshold smaller than $\frac12$, there will be, with high probability, a single Houdayer cluster spanning the entire graph.
    
    It remains unclear whether one can avoid the formation of such spanning Houdayer clusters in general. Some research along this line has been carried out in the case of the Ising model without linear terms ($h_i=0$) and proposes to solve this problem by restricting the use of Houdayer moves to low temperatures~\cite{cluster_algo_any_space_dimension}. Other works consider artificially increasing the percolation threshold by considering random subgraphs of the original problem graph~\cite{cluster_move_diluted_spins, cluster_move_diluted_spins_critical_behavior}.
    
    \item \textit{The Houdayer move is biased}: Unlike other cluster algorithms like the Wolff or Swendsen-Wang algorithm~\cite{wolff, WANG1990565}, the choice of cluster for the Houdayer move does \textit{not} depend on the weights $J_{ij}$ of the pairwise interactions between variables, but solely on the topology of the graph $G$. During the cluster selection process one first chooses a site with local overlap value $q_i=-1$ and then grow the corresponding cluster. The probability of selecting a particular cluster is directly proportional to its size, thus biasing the move towards selecting larger clusters. Indeed, if the Houdayer clusters are given by $\mathcal{C} = \{C_1, C_2, \dots, C_k\}$, then the probability of selecting cluster $C_i\in \mathcal{C}$ is equal to 
    \begin{equation}\label{equation:original_houdayer_cluster_probability}
        \mathbb{P}(C_i) = \frac{|C_i|}{\sum_{j=1}^k |C_j|}.
    \end{equation}
    There is however no clear reason as to why a larger number of flips should be more desirable in general.
\end{itemize}
In the next section, we resolve some of these issues by extending the Houdayer move, allowing a more uniform selection across clusters together with the possibility to select multiple of them at the same time.

\section{Generalized Houdayer Algorithm}\label{section:generalized_houdayer}
    Multiple variations of the Houdayer algorithm have been explored with the purpose of suppressing its drawbacks. While some of them focus on applying the move only in certain situations to improve its overall efficacy~\cite{cluster_algo_any_space_dimension}, others directly change the inner workings of the algorithm. An example of the latter uses the idea that Houdayer clusters can be grown in a similar way as in the Swendsen-Wang-Wolff algorithm~\cite{cluster_move_diluted_spins, percolation_signature_spin_glass}. In these modified versions, unlike the original move in which connected sites with $q_i=-1$ form a Houdayer cluster, edges are allowed to be removed with a certain probability that depends on the interaction strengths $J_{ij}$. As a result, certain sites will be omitted from the cluster being grown. Despite allowing some form of control over the clusters sizes, these kinds of transformations do not preserve the energy. Here, we design a generalization which retains this property, while extending the possibilities offered by the original Houdayer move.

\subsection{Description and Properties}
In the following, given a pair of configurations $p\in \mathcal{X}\times\mathcal{X}$, we write $\mathcal{C}(p)$ to denote the set of Houdayer clusters defined by the local overlap of $p$. When it is clear from the context what the pair is, we just write $\mathcal{C}$. We now describe our extension of the Houdayer move, which works as follows:
\begin{enumerate}
    \item Given a pair of independent configurations of variables $(\boldsymbol{\sigma}^{(1)}, \boldsymbol{\sigma}^{(2)})\in \mathcal{X}\times\mathcal{X}$, compute the local overlap $q_i = \sigma_i^{(1)}\sigma_i^{(2)}$ at every site $i \in V$ just as in the original Houdayer move, which defines a set of Houdayer clusters $\mathcal{C}$.
     
     \item Select a combination $T \subseteq \mathcal{C}$ of Houdayer clusters according to some desired probability distribution $\nu_{\mathcal{C}} \in \mathcal{P}\left(2^{\C}\right)$ and flip the values of the corresponding variables in both configurations.
\end{enumerate}
Most importantly, the generalized move considers not only flipping a single Houdayer cluster, but any combination of them according to $\nu_{\mathcal{C}}$, including the cases of zero or all clusters. A depiction of the move can be found on the right side of \cref{fig:houdayer}. By enabling the choice of any probability measure over combinations of Houdayer clusters, this allows to control ``how far" we jump in terms of Hamming distance in the configuration space, which can be beneficial in certain applications. Note that each particular choice of $\nu_\C$ leads to a specific instance of our extension. We highlight below three such choices which are repeatedly used in this work and name them for further reference:
\begin{itemize}
    \item \textit{Original Houdayer move}: Denoting by $\mathcal{T}_1 = \{T \subseteq \C: |T|=1\}$ the set of combinations containing exactly one Houdayer cluster, one can recover the original Houdayer move described in \cref{section:description_houdayer} (see \cref{equation:original_houdayer_cluster_probability}) by setting
    \begin{equation}\label{equation:original_houdayer_distribution}
        \nu_{\mathcal{C}}(T) = \mathbbm{1}_{\mathcal{T}_1}(T) \cdot \frac{\sum_{C\in T} |C|}{\sum_{S\in \mathcal{T}_1}\sum_{C\in S} |C|}
    \end{equation}
    for $T\subseteq \C$.
    
    \item \textit{Uniform Cluster Combination (UCC) Houdayer move}: There is a priori no specific reason to restrict ourselves to some particular distribution $\nu_{\mathcal{C}}$ in general. Thus, choosing the uniform distribution
    \begin{equation}\label{equation:extended_houdayer_distribution}
        \nu_{\mathcal{C}}(T)=\frac{1}{2^{|\mathcal{C}|}}
    \end{equation}
    with $T\subseteq \C$ gives access to all accessible pairs without any bias.
    
    \item \textit{Uniform Hamming Distance (UHD) Houdayer move}: Sometimes, it can be desirable to allow biased choices during the cluster selection, for example when performing Markov chain monte-Carlo sampling (see \cref{subsection:improved_sampling}). For a given a set of Houdayer clusters $\C$ with $\sum_{C\in \C}|C|=l$, denote by $\mathcal{S}_i = \{T\subseteq\C : \sum_{C\in T}|C| = i\}$ the set of combinations which cover exactly $i$ variables for $0\leq i\leq n$. The distribution 
    \begin{equation}\label{equation:uhd_distribution}
        \nu_{\C}(T) = \frac{1}{(l+1)\cdot|\mathcal{S}_{t}|},
    \end{equation}
    where $t\triangleq\sum_{C\in T}|C|$, enables to visit configurations uniformly in terms of Hamming distance. A detailed treatment of the motivation behind this choice and its implementation is given in \cref{appendix:optimal_mcmc_houdayer}.
\end{itemize}

In the rest of this work, we refer to our extension as the \textit{generalized Houdayer move}, and we also use the term ``Houdayer moves" to refer to any move resulting from it. The following lemmata show that Houdayer moves also display the central property that the total function value remains constant, and that a number of characteristics of the original move can be extended to the more general case.
\begin{lemma} \label{lemma:extended_houdayer_same_energy}
Consider a pair of configurations of variables $(\boldsymbol{\sigma}^{(1)}, \boldsymbol{\sigma}^{(2)})$ and suppose that a generalized Houdayer move results in the new pair $(\boldsymbol{\tau}^{(1)}, \boldsymbol{\tau}^{(2)})$. Then we have $\mathcal{H}(\boldsymbol{\sigma}^{(1)}) + \mathcal{H}(\boldsymbol{\sigma}^{(2)}) = \mathcal{H}(\boldsymbol{\tau}^{(1)}) + \mathcal{H}(\boldsymbol{\tau}^{(2)})$.
\end{lemma}
\begin{proof}
See Appendix \ref{appendix:proof_same_energy_lemma}.
\end{proof}
\begin{lemma}\label{lemma:extended_houdayer_properties}
Consider a pair of configurations of variables $\left(\boldsymbol{\sigma}^{(1)}, \boldsymbol{\sigma}^{(2)}\right)$ such that their local overlap induces the set of Houdayer clusters $\mathcal{C}$. Then, we have that:
\begin{enumerate}
    \item There are $2^{|\mathcal{C}|}$ unique pairs of configurations (including the initial one) that can be reached.
    
    \item \label{property:extended_houdayer_reverse}For any pair $\left(\boldsymbol{\tau}^{(1)}, \boldsymbol{\tau}^{(2)}\right)$ that is reachable via some combination of Houdayer clusters $T\subseteq \mathcal{C}$, the reversed pair $\left(\boldsymbol{\tau}^{(2)}, \boldsymbol{\tau}^{(1)}\right)$ is also reachable, using the ``complement" combination $T^{\mathsf{c}} \triangleq \mathcal{C} \setminus T \subseteq \mathcal{C}$.
\end{enumerate}
\end{lemma}
\begin{proof}
See Appendix \ref{appendix:proof_extended_houdayer_properties}.
\end{proof}
\begin{remark}
Note the following:
\begin{enumerate}
    \item If the order of the pair is irrelevant for a certain application, then \cref{property:extended_houdayer_reverse} in \cref{lemma:extended_houdayer_properties} implies that half of the pairs that are reachable can be considered identical. To find all of them, one only needs to consider half of all possible combinations.
    \item \cref{property:extended_houdayer_reverse} in \cref{lemma:extended_houdayer_properties} provides extension of \cref{property:houdayer_reverse} in \cref{lemma:simple_houdayer_properties}, explaining in general why the choice of the combination $T = \mathcal{C}$ (which covers all the Houdayer clusters) amounts to exchanging configurations in the input pair. Specifically, this is due to the fact that the complement $T^{\mathsf{c}}$ of the combination $T = \mathcal{C}$ is the empty combination $T^{\mathsf{c}} = \mathcal{C} \setminus T = \emptyset$, which itself just results in the input pair.
\end{enumerate}
\end{remark}

\subsection{Computational Complexity}
Discovering Houdayer clusters in a given local overlap grants access to a number of pairs of configurations with the same total function value, and this can be easily achieved via a Breadth-First Search (BFS) algorithm with time complexity $O(|V|+|E|)$~\cite{intro_to_algorithms}. After that, the choice of an accessible pair according to distribution $\nu_{\C}$ needs to be made. This can be computationally challenging since the number of such pairs is exponentially large in the number of Houdayer clusters. It turns out that for many distributions $\nu_{\C}$, in particular our three recurring examples, only few computational resources are necessary:
\begin{itemize}
    \item \textit{Original Houdayer move}: The distribution $\nu_{\mathcal{C}}(T) = \mathbbm{1}_{\mathcal{T}_1}(T) \cdot \frac{\sum_{C\in T} |C|}{\sum_{T\in \mathcal{T}_1}\sum_{C\in T} |C|}$ is achieved by performing the original Houdayer move. Since this only requires to grow one of the Houdayer clusters, the time complexity is $O(|V|+|E|)$ in the worst case, the same as a BFS algorithm.
    \item \textit{UCC Houdayer move}: The distribution $\nu_{\mathcal{C}}(T) = \frac{1}{2^{|\mathcal{C}|}}$ is also achieved efficiently. After discovering all Houdayer clusters, a uniformly random combination of them can be chosen by considering each cluster independently and selecting it with probability $1/2$. \iffalse After discovering all Houdayer clusters, one can choose a uniformly random combination of them through the following steps:
        \begin{enumerate}
            \item Pick randomly the number of Houdayer clusters to be part of the combination, according to a binomial distribution. Specifically, denoting by $N$ the random variable representing this number, we use $N\sim \mathrm{Bin}(|\mathcal{C}|, \frac12)$, so that $\mathbb{P}(N=i) = \frac{\binom{|\mathcal{C}|}{i}}{2^{|\mathcal{C}|}}$.
            \item Sample uniformly at random $N$ Houdayer clusters from $\mathcal{C}$.
        \end{enumerate}
    To verify that it works correctly, consider any combination $T\subseteq \mathcal{C}$. In order for it to be selected through the process defined above, we must first have that $N=|T|$ and then, out of all combinations of size $|T|$, choose precisely $T$. Formally, the resulting distribution $\nu_\C$ is
    \begin{align*}
        \nu_{\mathcal{C}}(T) &= \mathbb{P}(N = |T| \cap \text{ select } T) \\
        &= \mathbb{P}(N = |T|) \cdot \mathbb{P}(\text{select } T | N=|T|) \\
        &= \frac{\binom{|\mathcal{C}|}{|T|}}{2^{|\mathcal{C}|}} \cdot \frac{1}{\binom{|\mathcal{C}|}{|T|}} \\
        &= \frac{1}{2^{|\mathcal{C}|}},
    \end{align*}
    as expected.\fi Hence, the running time of this move is dominated by the time required to compute all Houdayer clusters, which is $O(|V|+|E|)$. We note that in this case, the resulting algorithm is similar to the one presented in~\cite{graphical_representation_cluster_algo}.
    \item \textit{UHD Houdayer move}: The distribution in \cref{equation:uhd_distribution} is achieved by first computing all Houdayer clusters. After that, one needs to create and fill a table via dynamic programming (details can be found in \cref{appendix:optimal_mcmc_houdayer}). In the worst-case, the number of Houdayer clusters is equal to $|V|$ (if the graph is the empty graph), so that filling the table takes time $O(|V|^2)$. Then, retrieving a combination from the table takes time $O(|V|)$, leading to an overall running time of $O(|V|^2)$.
\end{itemize}

\subsection{Expected Number of Flips}\label{section:expected_number_of_flips}
% TODO: change the beginning to not repeat what is in Applications
We take a short detour to analyze the number of flips performed during a Houdayer move. This can be useful in order to provide fair comparisons between methods that uses Houdayer moves and ones that do not. When studying Monte-Carlo dynamics for example (see \cref{subsection:improved_sampling}), such comparisons can be achieved when said methods flip on average the same number of variables.

First, a simple but crucial question to ask is: given a pair of configurations, how is the number of flips calculated for the generalized Houdayer move? For example, defining the number of flips to be the number of sites covered by the chosen combination of Houdayer clusters would not be representative of how the dynamics evolves. Indeed, in the case where all Houdayer clusters are selected, the move effectively performs zero flips since it only exchanges the configurations in the pair. Instead, since \cref{property:extended_houdayer_reverse} in \cref{lemma:extended_houdayer_properties} more generally states that choosing a combination $T\subseteq \C$ or its complement $T^c\subseteq \C$ essentially results in the same configurations at the output, the true number of flips performed via that combination should be taken as
\begin{equation}\label{equation:true_number_flips}
    \eta_{\C}(T) \triangleq \min\left(\sum_{C\in T}|C|, \sum_{C\in T^{\mathsf{c}}}|C|\right),
\end{equation}
i.e., the minimum number of flips between $T$ and $T^{\mathsf{c}}$. Hence, given a pair $p\in\XX$, the expected number of flips becomes
\begin{equation*}
    \eta(p) \triangleq \mathbb{E}_{\nu_{\C(p)}}[\eta_{\C(p)}] = \sum_{T\subseteq \C(p)}\nu_{\C(p)}(T) \cdot \eta_{\C(p)}(T).
\end{equation*}
With this, we find that the expected number of flips performed by a generalized Houdayer move when pairs appear according to some distribution $\mu \in \mathcal{P}(\XX)$ is given by
\begin{equation}\label{equation:expected_number_flips}
    \eta_{\mu} \triangleq \mathbb{E}_{\mu}[\eta].
\end{equation}
Since $\mu$ might be dependent on the problem at hand, we restrict ourselves to the special case where a uniformly random pair of configurations is assumed to be given as input, effectively considering $\mu(p)=\frac{1}{2^{2|V|}}$ for $p\in \XX$. In this case, as already noted in Section \ref{section:houdayer_limitations}, the assignment of local overlap values $q_i$ becomes equivalent to a site percolation process, so that the study of number of flips morphs into the study of the percolation process itself. Such an analysis can easily become involved even for simple graph topologies. We illustrate this by comparing the expected number of flips of different versions of the Houdayer move for the 1D graph.
\begin{example}[1D Graph]\label{example:1d_graph}
Let $G=(V, E)$ with $V=\{1, \dots, |V|\}$ and $E=\{(i, i+1): i=1, \dots, |V|-1\}$. Then, the site percolation process can be solved analytically when assuming $|V| \to +\infty$ (see Appendix \ref{appendix:1d_graph}), which allows to compute the expected number of flips in a number of cases:
\begin{itemize}
    \item Using the original Houdayer move yields
    \begin{equation}
        \eta_{\mu} = 3.
    \end{equation}
    \item Using the UCC Houdayer move yields
    \begin{equation}
        \frac{|V|}{4}\left(1-\sqrt{\frac{6}{|V|}}\right) \leq \eta_{\mu} \leq \frac{|V|}{4},
    \end{equation}
    but the exact closed-form expression (see Appendix \ref{appendix:1d_graph}) can be evaluated numerically.
\end{itemize}
This shows that the original Houdayer move flips 3 variables on average, while the version of the Houdayer move which considers all possible combinations uniformly at random flips $\Theta(|V|)$ variables on average. To see how precise these closed-form expression are in practice for finite graphs, we approximate for increasing values of $|V|$ the expected number of flips by sampling uniformly random pairs of configurations and applying Houdayer moves while keeping track of the sample average of the number of flips performed. The results depicted in Fig. \ref{fig:expected_number_flips} display that the closed-form expressions match the actual behavior well even when $|V|$ is small.
\begin{figure}
    \centering
    \includegraphics[width=\columnwidth]{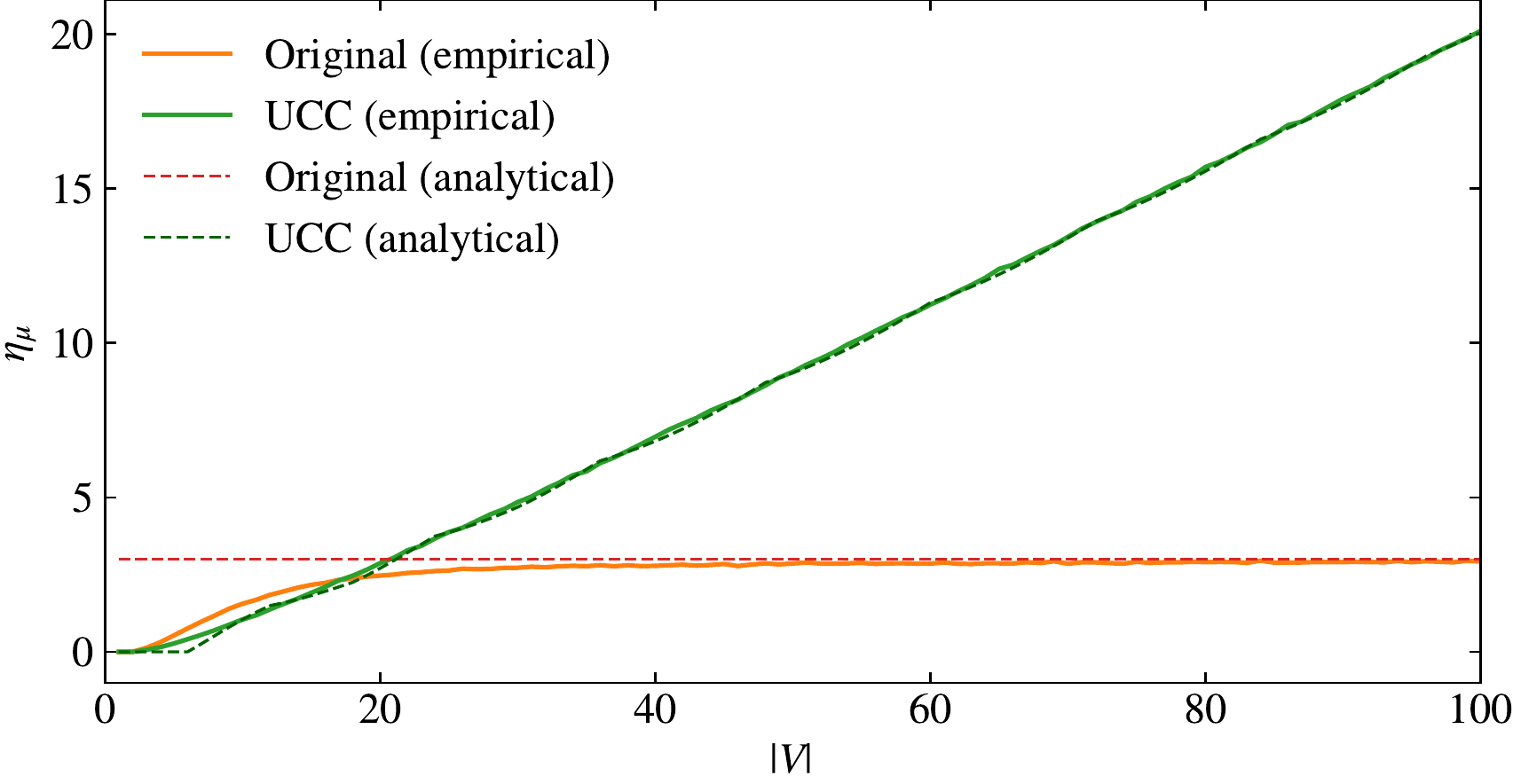}
    \caption{Comparison of the expected number of flips performed by the original and UCC version of the Houdayer move on the 1D graph as a function of $|V|$ to illustrate Example \ref{example:1d_graph}. The values of interest are estimated by sampling 1000 pairs of configurations uniformly at random and averaging for each move version the number of flips performed (solid lines). The expected number of flips given by the closed-form formulas derived in Appendix \ref{appendix:1d_graph} is also plotted (dashed lines).}
    \label{fig:expected_number_flips}
\end{figure}
\end{example}
Example \ref{example:1d_graph} reveals that even in a simple scenario, the number of expected flips performed by a generalized Houdayer move can vary quite extensively depending on the choice of $\nu_{\C}$. Hence, for general graphs, and in particular finite ones, finding a closed-form expression for the expected number of flips is far from trivial, rendering it difficult to determine what a fair comparison with other methods would entail.

\section{Applications} \label{section:applications}
In the previous section, we introduced the notion of the generalized Houdayer move and analyzed some of its key features and properties. In this section, we discuss how Houdayer moves can be used in the contexts of sampling from a probability distribution, and optimization.

    \subsection{Improved Markov Chain Monte-Carlo Sampling} \label{subsection:improved_sampling}
        The problem of sampling from a distribution is of paramount importance in a multitude of disciplines. However in high dimensions, usual sampling strategies such as rejection sampling or importance sampling become inefficient due to the curse of dimensionality \cite[Chapter 29]{MacKay2003}. The procedures of choice in this case are Markov chain Monte-Carlo methods, and in particular the Metropolis-Hastings (MH) algorithm~\cite{metropolis_hastings}. As a reminder, the MH algorithm consists of building an ergodic (irreducible and aperiodic) Markov chain which has the desired distribution as its stationary and limiting distribution \cite{metropolis_hastings}. In practice, the process can be summarized in two steps which are iterated, starting from a random configuration $\boldsymbol{\sigma}$:
\begin{enumerate}
    \item Propose a transition $\boldsymbol{\sigma} \to \boldsymbol{\tau}$ to a new configuration according to some Markov chain with transition matrix $Q$.
    \item Accept or reject the new configuration $\boldsymbol{\tau}$ with probability $A_{\boldsymbol{\sigma} \to \boldsymbol{\tau}}$ following a certain acceptance scheme.
\end{enumerate}
Typically, the distribution of interest is the Boltzmann-Gibbs distribution stemming from the Hamiltonian $\mathcal{H}$, which is given by
\begin{equation}\label{equation:boltzmann_distribution}
    \mu_\beta(\boldsymbol{\sigma}) = \frac{e^{-\beta \mathcal{H}(\boldsymbol{\sigma})}}{Z_{\beta}}, 
\end{equation}
where $\beta > 0$ is a parameter called the inverse temperature and $Z_{\beta} = \sum_{\boldsymbol{\sigma}\in \X} e^{-\beta \mathcal{H}(\boldsymbol{\sigma})}$ is a normalization constant. The acceptance mechanism commonly used is the so-called Metropolis criterion given by $A_{\boldsymbol{\sigma} \to \boldsymbol{\tau}}=\min\left(1, \frac{\mu_\beta(\boldsymbol{\tau})\cdot Q_{\boldsymbol{\tau}\to \boldsymbol{\sigma}}}{\mu_\beta(\boldsymbol{\sigma})\cdot Q_{\boldsymbol{\sigma}\to \boldsymbol{\tau}}}\right)$, which in our setting reduces to 
\begin{equation}\label{equation:metropolis_acceptance}
    A_{\boldsymbol{\sigma} \to \boldsymbol{\tau}}=\min(1, e^{-\beta (\mathcal{H}(\boldsymbol{\tau})-\mathcal{H}(\boldsymbol{\sigma}))})
\end{equation}
when $Q$ is symmetric. Other acceptance schemes can be used in order to tune the acceptance rate of the process~\cite[Section 11.2.1]{markov_chains_bremaud}.
\begin{remark}
It is well known that when $\beta \to \infty$, the corresponding Boltzmann-Gibbs distribution is nothing but $\mu_{\infty}(\config) = \frac{\mathbbm{
1}_{\calH_0}(\config)}{|\calH_0|}$, i.e., the uniform distribution over the set of configurations that minimize $\calH$~\cite[Eq. (2.17)]{information_physics_computation}. In such a case, instead of directly running the MH algorithm with infinite $\beta$, one typically uses simulated annealing~\cite{simulated_annealing_paper}, a process which starts at low values of $\beta$ and each iteration sees its value increase. After enough steps, $\beta$ is deemed large enough so that the distribution sampled is close to $\mu_{\infty}$.
\end{remark}
Two crucial properties determine the capacity of the MH algorithm to correctly sample from the target distribution in~\cref{equation:boltzmann_distribution}:
\begin{itemize}
    \item \textit{High-Dimensionality}: Regardless of the algorithm employed, sampling from a high-dimensional distribution is a particularly demanding task because the space to explore becomes extremely large. Obtaining good samples requires to explore most if not all the configurations in the space to find regions with non-zero probability.
    
    \item \textit{Complexity of $\mathcal{H}$}: As $\beta$ increases, the distribution concentrates more and more on configurations with lower values of $\mathcal{H}$ while the support of the distribution itself shrinks to only those configurations, resulting in multiple disjoint regions with non-zero probability. For example, when $\mathcal{H}$ contains multiple minima around the same value, this can result in a multimodal distribution which is extremely hard to sample as it can be difficult to find and accept distant transitions leading to other modes of the distribution.
\end{itemize}
In summary, a good MH algorithm requires transitions that allow the exploration of configurations both close to and far from the current one in terms of Hamming distance, in order to avoid being trapped in a particular neighborhood. At the same time, transitions need to be accepted to be successful. This requires $\mathcal{H}(\boldsymbol{\tau})-\mathcal{H}(\boldsymbol{\sigma})$ to be either negative or a small positive number since the probability of acceptance decays exponentially fast with respect to it. Thus, along with the usage of the MH comes the quest of finding the most suitable transition proposals that will enable the mixing process to be done as fast as possible.

One of the most widely used and analyzed (mainly due to its meager computational requirements) is the so-called \textit{single-flip} strategy, which proposes to flip the value of one variable chosen uniformly at random. The issue with this approach is that samples tend to be correlated due to the locality of the move. In particular for $\beta \gg 0$, it can be extremely difficult to discover new regions with non-zero probability since after reaching a local minimum of $\mathcal{H}$, all proposals will be rejected since local moves can only increase the function value.

\subsubsection{Sampling and Houdayer moves}
We now explain how one can use Houdayer moves within the MH algorithm in order to benefit from their non-locality.

Since the Houdayer moves work with pairs of configurations, we need to redefine the target distribution so that it assigns probability to elements in $\X\times\X$. We thus choose to employ the joint distribution over the space of independent pairs of configurations $\X\times\X$ at the same inverse temperature $\beta$, given by 
\begin{align*}
    \mu^{\text{pairs}}_\beta\left(\left(\boldsymbol{\sigma}^{(1)}, \boldsymbol{\sigma}^{(2)}\right)\right)& = \mu_\beta\left(\boldsymbol{\sigma}^{(1)}\right) \cdot \mu_\beta\left(\boldsymbol{\sigma}^{(2)}\right) \\
    &= \frac{e^{-\beta (\mathcal{H}(\boldsymbol{\sigma}^{(1)}) + \mathcal{H}(\boldsymbol{\sigma}^{(2)}))}}{Z^{\text{pairs}}_{\beta}},
\end{align*}
where $Z^{\text{pairs}}_{\beta}$ is a normalization constant. Alternatively, one can view this as the usual Boltzmann-Gibbs distribution derived from the total Hamiltonian $\mathcal{H}^{\text{pairs}}\left(\left(\boldsymbol{\sigma}^{(1)}, \boldsymbol{\sigma}^{(2)}\right)\right) = \mathcal{H}(\boldsymbol{\sigma}^{(1)}) + \mathcal{H}(\boldsymbol{\sigma}^{(2)})$.

We can now apply the framework of the MH algorithm and use Houdayer moves to propose transitions. By exploring the space of pairs of configurations while staying at a same overall function value, Houdayer moves offer a unique exploration scheme. Specifically, by allowing transitions to new pairs of configurations at the same total function value and the same temperature, the move is guaranteed to be accepted. In addition, its ability to reach pairs which can be substantially different in terms of Hamming distance can be highly beneficial when the process is stuck in some region of the space.

An important point with this approach is that the Houdayer move on its own does not make the process ergodic, since the resulting Markov chain is not irreducible. Indeed, starting from some initial pair of configurations and applying only the Houdayer move, all new pairs visited will have the same value of $\mathcal{H}^{\text{pairs}}$. Clearly, this will not result in an exploration of the entire space $\X\times\X$ in general. Thus, one needs an extra step to render the whole scheme ergodic, the simplest one being to additionally perform single-flip moves on each configuration in the pair independently. In Section \ref{subsection:improved_sampling_experiment}, we explicitly demonstrate that using such a combined strategy with the generalized Houdayer move improves sampling.
\begin{remark}
    One can more generally start with $R\geq 2$ configurations at the same inverse temperature $\beta$. The process remains similar, and two steps are performed at each iteration. First, the single-flip strategy is applied on each one of the $R$ configurations independently. Then, they are paired together so that Houdayer moves can be performed on each created pair.
\end{remark}
Note that an issue arises when comparing the usual single-flip procedure with the one augmented with Houdayer moves. Since Houdayer moves can only help to reach the target distribution faster, it would not be fair to compare the two dynamics directly. One way to solve it is to make both transformations propose to flip the same number of variables on average. On one hand, the single-flip procedure ensures that all move proposals change the value of exactly one variable, while on the other hand, Houdayer moves can in principle flip the value of all $|V|$ variables. Even though this might indicate to compare one application of a Houdayer move to $|V|$ single-flips, we know from \cref{section:expected_number_of_flips} that Houdayer moves can in reality flip widely different number of variables. In fact, estimating the expected number of flips they induce is a difficult task since it depends on the topology of the graph $G$, the distribution $\mu$ over pairs, and the distribution $\nu_{\C}$. In the context of sampling from the Boltzmann-Gibbs distribution via the MH algorithm, since the pairs encountered become increasingly distributed according to the target distribution $\mu^{\text{pairs}}_{\beta}$, the expected number of flips after enough iterations is given by $\eta_{\mu_{\beta}^{\text{pairs}}}$. Since evaluating this quantity requires to sample from $\mu_{\beta}^{\text{pairs}}$, the problem we started with, this cannot be determined prior to running the MH algorithm. 

Although it is not possible to easily determine the expected number of flips performed by Houdayer moves in general, one can still find ways to increase the fairness when comparing it to single-flips. In practice, a commonly used scheme is to apply one Houdayer move every $|V|$ single-flips instead of applying it after every single-flip.

    \subsection{Low-Energy Configurations Generation: A Genetic Algorithm} \label{subsection:genetic_algorithm}
        We now present a genetic algorithm leveraging Houdayer moves that can be easily customized to perform several tasks, most of which are used to improve the quality of existing solutions. A similar idea has been explored in \cite[Section II.A]{quantum_assisted_genetic_algorithm} to yield a quantum-assisted genetic algorithm, but unnecessarily restricts the set of reachable pairs by using the original form of the Houdayer move. Besides, the generality of our procedure will allow us to recover the algorithm presented in \cite[Section II.A]{icm_generating_more_optimum}.

There exist various ways to formulate genetic algorithms, the core idea being that they provide a procedure which maintains a population of configurations (called individuals) which are usually initialized by drawing random configurations from the solution space. The population evolves in a way which attempts to imitate evolution, ideally tending towards configurations of low function value \cite{quantum_assisted_genetic_algorithm}. The three main steps that the algorithm iterates over are (see \cref{fig:genetic_algorithm}):
\begin{enumerate}
    \item \textit{Crossover}, in which individuals are combined to produce offspring with attributes derived from the parent individuals. Here, we perform this step by pairing configurations and considering all reachable pairs stemming from their local overlap.
    
    \item \textit{Mutation}, in which individuals go through random transformations which produce new individuals.
    
    \item \textit{Selection}, in which individuals that are the fittest (according to some criteria) are retained while others are removed in order to regulate the size of the population.
\end{enumerate}
The simplicity of this heuristic method makes it easy to implement in practice, and thus offers an attractive choice when solving optimization problems.

\subsubsection{Algorithm Description}\label{section:ga_description}
\begin{algorithm}
\SetKwInOut{Input}{Input}
\caption{Genetic Algorithm for Improved Sampling}\label{alg:genetic_algorithm}

\Input{Number of iterations $N$, Initial population $P$, Threshold $k$ for the population size limit}

\For{$n \in \{1, \dots, N\}$}{

    Construct a set of pairs of configurations $R$ from $P$.\\
    \For{$r \in R$}{
        Compute the local overlap and generate a set  $C$ of new pairs of configurations. \\
        Apply the mutation operator on each $c \in C$ with some probability, yielding a set of configurations $M$.\\
        $P \gets P \cup M$\\
    }
    Retain $k$ configurations from $P$ with lowest energy values.\\ 
}
\Return $P$
\end{algorithm}
The procedure is described in a general way in \cref{alg:genetic_algorithm}, and simply follows the steps of any genetic algorithm while using some of the tools developed in previous sections. It possesses a number of parameters and degrees of freedom which make the scheme suitable in a variety of situations. Below, we consider how the modification of these parameters affects the overall running time.
\begin{figure*}[ht]
\centering
\includegraphics[width=\textwidth]{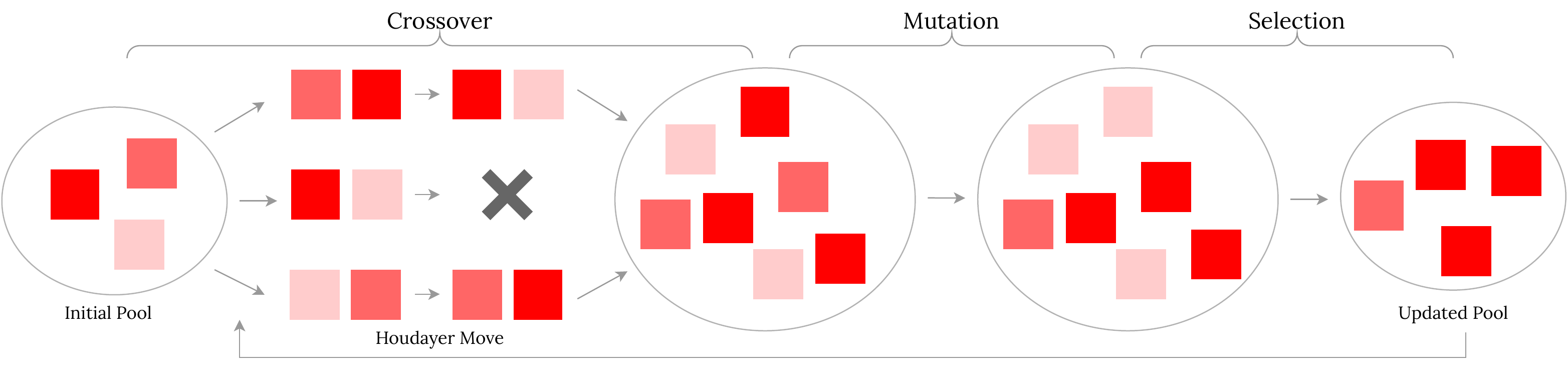}
\caption[]{Visual representation of Algorithm \ref{alg:genetic_algorithm}, which starts from an initial pool of configurations, and iterates over three distinct steps: crossover, mutation, and selection (see Section \ref{subsection:genetic_algorithm}). Darker shades of red indicate a lower function value.}
\label{fig:genetic_algorithm}
\end{figure*}
\begin{itemize}
    \item Constructing the set of pairs $R$ given the current pool of configurations can be subject to many modifications which alter the overall running time of the method. We briefly mention some of them here. Given a pool $P$, one can consider all possible $\binom{|P|}{2}=O(|P|^2)$ pairs of distinct configurations, however this can become costly if the pool size is large. One can instead randomly pair configurations in the pool, yielding $\lfloor \frac{|P|}{2} \rfloor = O(|P|)$ pairs. Finally, one could instead sample a constant number of random pairs from the pool.
    
    \item The crossover step may also be modified in different ways. Once the local overlap is computed for a pair, there is freedom in choosing how new pairs should be reached. One can either follow the steps of the original Houdayer move, focusing on pairs arising from flipping one Houdayer cluster at a time; or, one could also consider any combination of Houdayer clusters as in the generalized version of the move. In fact, to ensure that excessive resources are not required at this step (which could happen if the number of Houdayer clusters is large), one can also sample a constant number of combinations of Houdayer clusters via a desired probability distribution $\nu_{\C}$ to limit the number of computations at this step.
    
    \item The mutation phase of the algorithm can also be modified depending on the problem at hand. For example, for a given configuration, one simple scheme is to flip the value of one of the variables in it with some probability.
    
    \item Finally, the selection process can be altered as desired. Here, the number of configurations $k$ retained at each iteration dictates the maximum size of the pool during the algorithm. If no selection process is done, the pool just keeps growing, slowing down the rest of the procedure.
\end{itemize}
Depending on these choices, one can achieve a polynomial running time. For example, if one limits the size of the pool to $k$ configurations, uses the simple mutation scheme which flips one variable with some probability, computes all pairs arising from the pool at each iteration, and moreover computes new pairs by only flipping one Houdayer cluster at a time during the crossover phase, the overall running time becomes $O(Nk^2|V|^2)$.
\begin{remark}
The way in which configurations are constructed in the genetic algorithm highlights an important aspect of Houdayer moves. Even though their main application is for sampling purposes, which explains why accessible pairs are chosen making a random choice via $\nu_{\C}$, one can leverage the move in a deterministic manner. That is, if the goal is to produce many configurations, there is no need to repeatedly choose only one of the many possible pairs at random each time. Instead, multiple or even all of them can be computed at once by considering the corresponding combinations of Houdayer clusters.
\end{remark}
In the next two sections, we show some applications of \cref{alg:genetic_algorithm} beyond its potential use as an approximate solver for optimization problems. The scheme can in fact be seen as a way to enhance and diversify an existing pool of configurations. Indeed, when starting with an initial pool of configurations with low value of $\mathcal{H}$, the procedure will continue to generate such configurations.

\subsubsection{Feasible Space Exploration in Constrained Optimization Problems}
Constrained optimization problems divide the space of configurations into feasible and non-feasible regions. The aim is to find a solution respecting all constraints (i.e., which belongs to the feasible region) while at the same time being optimal. For such problems, even getting in the feasible region can be a non-trivial task, so that finding a few feasible configurations is already enough. In an ideal scenario, one would like to be able to navigate the space of feasible configurations in order to carry out the minimization directly in this space without having to worry about the constraints anymore.

To see how Houdayer moves can be useful here, we first assume that the function $\mathcal{H}$ to minimize ``encodes" feasibility. By that, we mean that it acts as a penalty function, assigning infinite values to infeasible configurations and finite values to feasible ones. Formally, we have
\begin{equation*}
    \mathcal{H}(\boldsymbol{\sigma})  \begin{cases}
       < +\infty, &\quad\text{if $\boldsymbol{\sigma}$ is feasible}\\
       = + \infty, &\quad\text{otherwise}.
     \end{cases}
\end{equation*}
It is easy to see that if one has a pair of feasible configurations, a successful application of the Houdayer move can only lead to another pair of feasible configurations. In general, starting with an initial pool with feasible configurations, new ones can be reached using \cref{alg:genetic_algorithm}, making it a computationally attractive procedure to generate feasible configurations.

While this explanation helps building intuition, it breaks down in practice because of the form of $\mathcal{H}$ requiring coefficients $J_{ij}, h_i$ to be finite (see \cref{equation:problem_statement}). We thus relax our assumption and instead suppose that $\mathcal{H}$ encodes feasibility in the following way: there is a finite value $\lambda \in \mathbb{R}$ that we call \textit{feasibility threshold} such that for any configuration $\boldsymbol{\sigma}\in \mathcal{X}$,
\begin{equation*}
    \mathcal{H}(\boldsymbol{\sigma})  \begin{cases}
       < \lambda, &\quad\text{if $\boldsymbol{\sigma}$ is feasible}\\
       \geq \lambda, &\quad\text{otherwise}.
     \end{cases}
\end{equation*}
In many problems of interest, the feasibility threshold is known and can be easily tuned since constraints are encoded via penalty terms which can be scaled as desired \cite{ising_formulation_np_problems}. In order to retain the ability to find other feasible configurations using the Houdayer move, one needs the following result, which is illustrated in Fig.~\ref{fig:feasibility_threshold}.
\begin{lemma}\label{lemma:constrained_optimization_condition}
 Let $\mathcal{H}: \mathcal{X}\to\mathbb{R}$ be a function corresponding to a constrained minimization problem, with minimum value $\mathcal{H}_0$ and feasibility threshold $\lambda$. Then for any pair of configurations $\left(\boldsymbol{\sigma}^{(1)}, \boldsymbol{\sigma}^{(2)}\right)\in \mathcal{X}\times\mathcal{X}$ with
 \begin{equation}\label{equation:feasibility_threshold_condition}
     \mathcal{H}\left(\boldsymbol{\sigma}^{(1)}\right) +\mathcal{H}\left(\boldsymbol{\sigma}^{(2)}\right) < \lambda+\mathcal{H}_0,
 \end{equation}
 any configuration resulting from the application of a Houdayer move on $\left(\boldsymbol{\sigma}^{(1)}, \boldsymbol{\sigma}^{(2)}\right)$ is feasible.
\end{lemma}
\begin{proof}
Since the total function value is preserved through a Houdayer move, it suffices to check that $\mathcal{H}\left(\boldsymbol{\sigma}^{(1)}\right) +\mathcal{H}\left(\boldsymbol{\sigma}^{(2)}\right) < \lambda + \mathcal{H}_0$ implies $\mathcal{H}\left(\boldsymbol{\sigma}^{(i)}\right) < \lambda$ for $i=1, 2$, where $\left(\boldsymbol{\sigma}^{(1)}, \boldsymbol{\sigma}^{(2)}\right) \in \mathcal{X}\times\mathcal{X}$. We show that the contrapositive holds, proving the result.

Suppose w.l.o.g. that $\mathcal{H}\left(\boldsymbol{\sigma}^{(1)}\right) = \lambda + \varepsilon \geq \lambda$ for some $\varepsilon \geq 0$. To have $\mathcal{H}\left(\boldsymbol{\sigma}^{(1)}\right) +\mathcal{H}\left(\boldsymbol{\sigma}^{(2)}\right) < \lambda + \mathcal{H}_0$, one needs $\mathcal{H}(\boldsymbol{\sigma}^{(2)}) < \mathcal{H}_0 - \varepsilon$, which is a contradiction since $\mathcal{H}_0$ is the minimum. Consequently we must have $\mathcal{H}\left(\boldsymbol{\sigma}^{(1)}\right) +\mathcal{H}\left(\boldsymbol{\sigma}^{(2)}\right) \geq \lambda + \mathcal{H}_0$, concluding the proof.
\end{proof}
\begin{figure}[!htbp]
    \centering
    \includegraphics[width=\columnwidth]{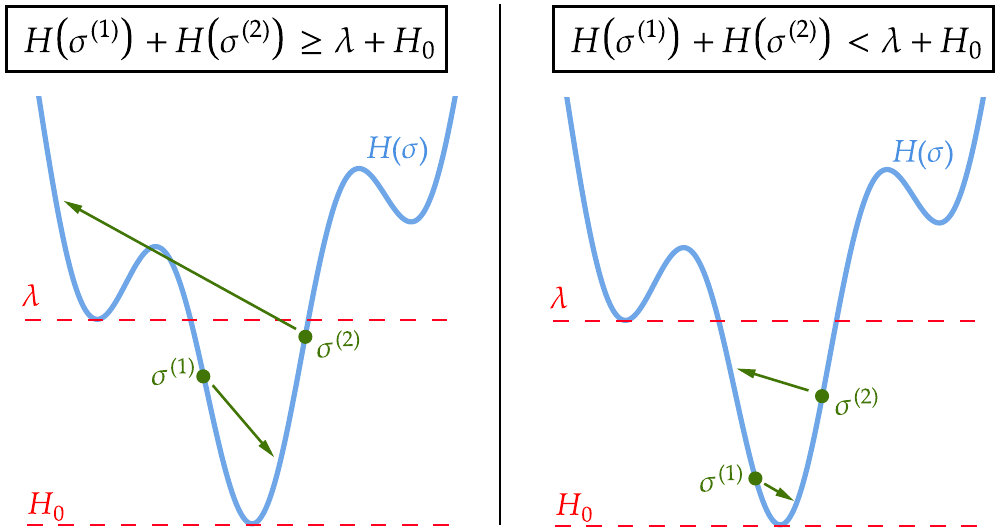}
    \caption{A visual representation of the necessary condition in constrained optimization problems to ensure that Houdayer moves yield feasible configurations, when feasible configurations are given (see \cref{lemma:constrained_optimization_condition}). Green dots represent configurations and green arrows show potential transitions given by the Houdayer move. The left plot depicts a scenario where the condition is not met and one of the resulting configuration becomes infeasible, while the right plot shows a scenario where the condition is met, ensuring that both configurations produced remain feasible. This illustrates in particular that the feasibility of configurations given as input (i.e., with energy lower than $\lambda$) does not ensure the feasibility of configurations produced.}
    \label{fig:feasibility_threshold}
\end{figure}
\begin{remark}
In practice since $\mathcal{H}_0$ is unknown, a simple trick is to modify the scale of penalty terms in $\mathcal{H}$ in order to artificially increase the feasibility threshold (see \cite{ising_formulation_np_problems}) to a large number so that $\lambda +\mathcal{H}_0 \approx \lambda$. The requirement from equation \cref{equation:feasibility_threshold_condition} then reads $\mathcal{H}\left(\boldsymbol{\sigma}^{(1)}\right) +\mathcal{H}\left(\boldsymbol{\sigma}^{(2)}\right) < \lambda$.
\end{remark}
To summarize, given a constrained optimization problem with known feasibility threshold, if one is able to build a pool of solutions such that any pair in the pool respects the condition in \cref{equation:feasibility_threshold_condition}, then it is ensured that all generated configurations are feasible.

\subsubsection{Ground-States Exploration in Degenerate Optimization Problems}\label{section:ground_state_generation}
Another use of \cref{alg:genetic_algorithm} is to generate optima of degenerate optimization problems using the following idea: when applying the Houdayer procedure, preserving the total function value typically means that the function value of one of the configuration decreases, while the other one increases by a similar amount. However when given two optimal configurations $\boldsymbol{\sigma}^{(1)}, \boldsymbol{\sigma}^{(2)} \in \mathcal{X}_0$, decreasing the value of any of them is impossible so that the two resulting configurations must preserve their individual function values and thus remain optimal.

As such, starting from a pool of optimal configurations (some subset of $\mathcal{X}_0$), one can potentially generate many more by simply using \cref{alg:genetic_algorithm} as follows. To ensure that the pool always contains as many optimal configurations as possible, the mutation and selection steps are removed. The pairing strategy and generation of new pairs during the crossover phase can still be adapted as highlighted in \cref{section:ga_description}, thus enabling to have a simple polynomial-time algorithm to generate optimal solutions.
\begin{remark}
It is important to note that this method does not necessarily allow to find all optimal configurations. Indeed, although Houdayer moves offer access to pairs of configurations at the same total function value, there is no certainty to visit all of them since the move is not ergodic.
\end{remark}
An illustration of the kind of results that can be achieved via this method is shown in \cref{subsection:maxcut_optima}.
        % First a Detour: Idea applied to Optimal configuration
        % Extend to generally low-energy solutions, application to constrained problems.

\section{Numerical Experiments} \label{section:numerical_results}
We now turn our attention to the numerical results obtained via the methods introduced in the previous section. Once again, our aim is to run the procedures in order to solve both sampling and optimization problems.

    \subsection{Improved Sampling for Markov Chain Monte-Carlo Methods}\label{subsection:improved_sampling_experiment}
    In order to understand how the use of the Houdayer moves can improve the sampling method mentioned in \cref{subsection:improved_sampling}, we use it to perform a derivative task, namely estimating mean values with respect to the Boltzmann-Gibbs distribution $\mu_{\beta}$ (see~\cref{equation:boltzmann_distribution}). Formally, given some function $f:\X \to \mathbb{R}$, we wish to evaluate the expected value
\begin{equation*}
    \mathbb{E}_{\mu_{\beta}}[f] = \sum_{\config \in \X} \mu_{\beta}(\config)\cdot f(\config),
\end{equation*}
which involves computing a costly sum over $2^{|V|}$ configurations. Instead, by using the Metropolis-Hastings algorithm, one can use the approximation (see \cite[Theorem 1.10.2]{markov_chains_norris})
\begin{equation}\label{equation:mcmc_approximation}
    \mathbb{E}_{\mu_{\beta}}[f] \approx \overline{f_n} \triangleq\frac{1}{n}\sum_{i=1}^n f(\config^{(i)}),
\end{equation}
where $n \ll 2^{|V|}$ is the number of iterations performed and $\config^{(1)}, \dots, \config^{(n)}$ are the configurations visited by the algorithm at each step. Hence, the quality of this approximation reflects how well the algorithm is able to sample from $\mu_{\beta}$. Note that such experiments have already been performed, showing the superiority of the original Houdayer move over the single-flip strategy~\cite{Houdayer_2001, cluster_algo_any_space_dimension}. The aim here is to compare the efficiency of the original version of the move against our generalization.

For this study, we additionally wrap the MH algorithm in a parallel tempering scheme~\cite{parallel_tempering} in order to improve the sampling results. This technique essentially performs the MH procedure at many different values of $\beta$ in parallel, while periodically allowing configurations to move across inverse temperatures. Coupled with Houdayer moves, the algorithm is able to explore energy landscapes both vertically (thanks to the multiple temperatures) and horizontally (thanks to the Houdayer move's ability to keep the total function value constant), making it particularly efficient~\cite{Houdayer_2001}. For each experiment in this section, one iteration corresponds to performing the single-flip strategy on each configuration at each temperature independently, followed by a Houdayer move on each pair of configurations at each temperature. Parallel tempering is often mentioned in the literature as being the ``state-of-the-art" when combined with Houdayer moves for solving general spin-glass problems~\cite{fujitsu_annealer}. Note that similar results can be obtained using simulated annealing instead of parallel tempering as the host algorithm, but the difference between the various Houdayer move types is less pronounced.

As a first step, we study ferromagnetic systems and examine the behavior of the average absolute magnetization, a quantity which is informative of whether a system has thermalized. Its true value is  given by $\mathbb{E}_{\mu_{\beta}}[m]$, where $m(\config) = \frac{\left|\sum_{i\in V}\sigma_i\right|}{|V|} $, but the evolution of $\overline{m_n}$ is studied instead via the parallel tempering procedure mentioned above. By looking at systems where the behavior of the average absolute magnetization is known, one can easily compare the different techniques we developed. In particular for ferromagnetic systems (i.e., with $J_{ij}>0$) at low temperature, neighboring variables want their values to be identical. This makes the average absolute magnetization tend to 1, so that we should equivalently observe $|\overline{m_n}-1|\to 0$ as $n$ grows, i.e., the relative error goes to 0. Multiple graph connectivities are examined, including the $100\times 100$ square lattice, the $70\times 70$ hexagonal lattice, and the $22\times 22\times 22$ cubic lattice, all with free boundary conditions and with bond values set to $J_{ij}=1$. The parallel tempering algorithm is run using twenty temperatures geometrically spaced from $\beta=0.1$ to $\beta=2$, but we only keep track of the evolution for $\beta=2$. We perform 1000 iterations for each variant of the move used, and the process is repeated 100 times starting from a uniformly random state in order to observe the average behavior.
\begin{figure*}[!htbp]
    \centering
    \includegraphics[width=\textwidth]{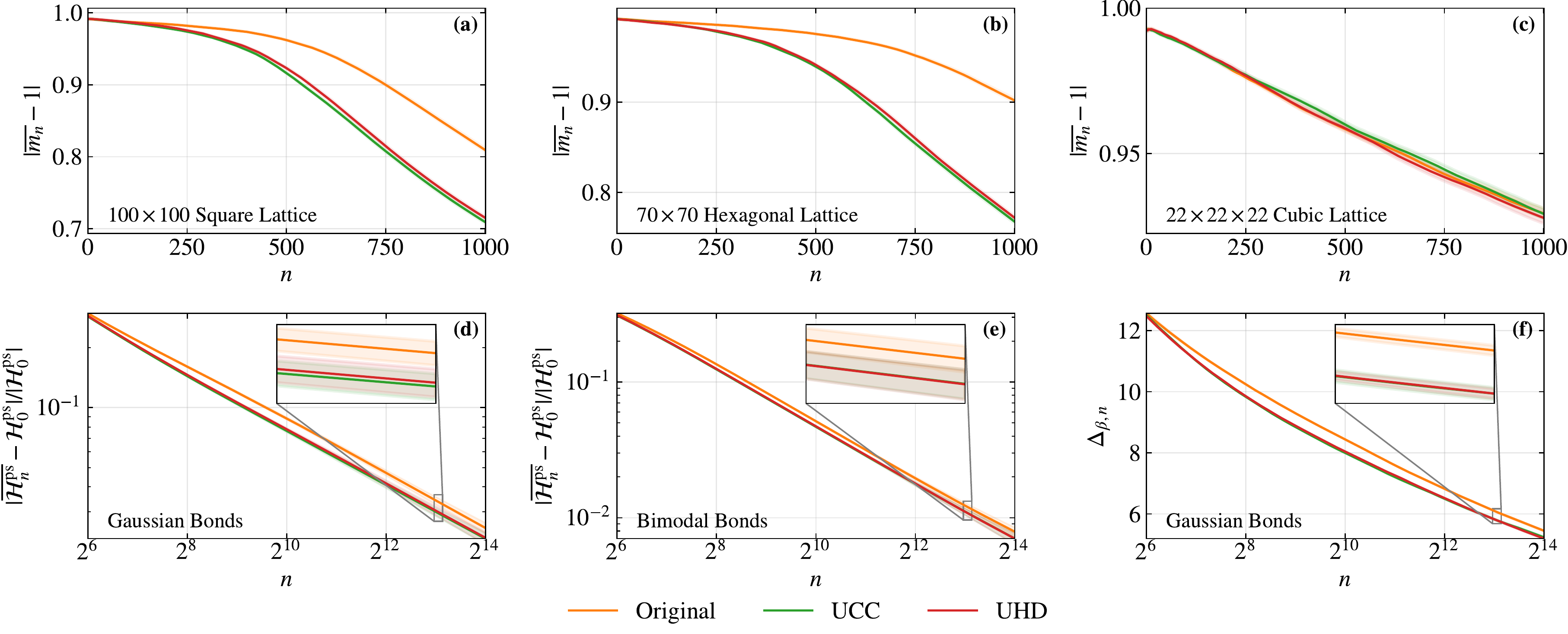}
    \caption{Results for the sampling experiments using Houdayer moves together with single-flips as Markov chain Monte-Carlo methods. (a)--(c) Evolution of the running average of the absolute magnetization through the relative errors $|\overline{m_n}-1|$ over $n=1000$ iterations when $\beta=2$. Each subfigure corresponds to a ferromagnetic system with a particular graph topology (square, hexagonal, or cubic lattices). The average over 100 random trajectories is shown in each case, together with the standard error of the mean in the shaded regions. (d)--(e) Evolution of the average energy per spin through the relative errors $\frac{|\overline{\mathcal{H}^{\mathrm{ps}}_n}-\mathcal{H}^{\mathrm{ps}}_0|}{|\mathcal{H}^{\mathrm{ps}}_0|}$ over $n=2^{14}$ iterations when $\beta=10$. Each subfigure corresponds to spin-glass systems with particular bond types (Gaussian or bimodal). (f) Evolution of the quantity $\Delta_{\beta, n}$ for spin-glass with Gaussian bonds. In plots (d)--(f), the graph topology is a lattice of size $100\times 100$ with periodic boundary conditions in both directions. In each case, the average over 10 random bonds assignments is shown, together with the standard error of the mean in the shaded regions.}
    \label{fig:magnetization_spin_glass}
\end{figure*}
The relative errors $|\overline{m_n}-1|$ are plotted in \cref{fig:magnetization_spin_glass}(a)--(c) as a function of the number of iterations $n$. One can see that there is a clear advantage in using moves resulting from our generalization for two-dimensional topologies, in particular for the hexagonal lattice. However for the cubic lattice, it is not possible to distinguish between the various methods, suggesting that the use of the original Houdayer move is already sufficient.

Next, we focus on the performance of Houdayer moves in the context of spin glasses, specifically for the Edwards-Anderson model~\cite[Chapter 2.1]{spin_glasses_and_information_processing}. In this setting, we have $h_i=0$ and the bonds form a $d$-dimensional lattice where the $J_{ij}$'s are distributed independently according to a probability distribution. Here, we use $d=2$ and consider two typical choices of distributions: the Gaussian case, where each $J_{ij}$ is distributed according to a centered Gaussian random variable with unit variance, and the bimodal case, where $J_{ij}$ takes values in $\{\pm 1\}$ with equal probability.

The first quantity we track is the average energy per spin. Its true value is given by $\mathbb{E}_{\mu_{\beta}}\left[\mathcal{H}^{\mathrm{ps}}\right]$, where $\mathcal{H}^{\mathrm{ps}}(\config)=\frac{\mathcal{H}(\config)}{|V|}$ (where the superscript ``ps" stands for ``per spin"), but the evolution of $\overline{\mathcal{H}^{\mathrm{ps}}_n}$ is studied instead. By choosing a large value of $\beta$, the energy observed should be close to that of ground states, for which precise values are known~\cite{ground_states_energy_ea_model, domain_wall_2d_lattice}. Specifically, when considering periodic boundary conditions in all directions in the two-dimensional lattice, the average energy per spin of ground-states can be approximated by $\mathcal{H}^{\mathrm{ps}}_0\approx-1.314787$ for Gaussian bonds, and $\mathcal{H}^{\mathrm{ps}}_0\approx-1.401922$ for the $\{\pm 1\}$ case ~\cite{domain_wall_2d_lattice}. In every case, the parallel tempering algorithm is run for $2^{14}$ iterations and using thirty temperatures, but we only keep track of the evolution for $\beta=10$. The plots in~\cref{fig:magnetization_spin_glass}(d)--(e) show the results for lattices of size $100 \times 100$ using each bond type, where the relative errors $\frac{|\overline{\mathcal{H}^{\mathrm{ps}}_n}-\mathcal{H}^{\mathrm{ps}}_0|}{|\mathcal{H}^{\mathrm{ps}}_0|}$ are reported. For both the Gaussian and bimodal bonds, the use of our generalization yields improved results compared to the original Houdayer move. However, there is no noticeable difference between the UCC and UHD Houdayer move. 
  
In the specific scenario where bonds follow a Gaussian distribution, one can rely on a quantity indicative of whether the system has thermalized~\cite{mcmc_spinglass_low_temp}. Defining the edge overlap between two configurations as $q^{\text{edge}}(\config^{(1)}, \config^{(2)}) = \frac{1}{|E|}\sum_{(i, j)\in E}\sigma_i^{(1)}\sigma_j^{(1)}\sigma_i^{(2)}\sigma_j^{(2)}$, we study the time-dependent behavior of
\begin{equation}\label{equation:thermalization_equation}
    \Delta_{\beta, n}(\mathcal{H}) \triangleq \frac{\overline{\mathcal{H}_n}}{|V|}-\beta\frac{|E|}{|V|}\left(1-\overline{q^{\text{edge}}_n}\right),
\end{equation}
which, for a fixed function $\mathcal{H}$ and an inverse temperature $\beta$, equates the internal energy per spin to the internal energy from the edge overlap~\cite{mcmc_spinglass_low_temp}. Additionally, since edge weights are themselves random (and thus $\mathcal{H}$ is random), we further take the expected value over this randomness. Hence, denoting by $\xi$ the probability distribution which produces functions $\mathcal{H}$ with the prescribed Gaussian edge weights, the quantity of interest is $\Delta_{\beta, n} \triangleq \mathbb{E}_{\xi}[\Delta_{\beta, n}(\mathcal{H})]$. A $n$ grows, one should observe that $\Delta_{\beta, n} \to 0$, indicating that the system is thermalized~\cite{correlations_pt_spinglass}. In our experiments, $\Delta_{\beta, n}$ is approximated by estimating $\Delta_{\beta, n}(\mathcal{H})$ on multiple instances with Gaussian edge weights and averaging the results obtained. Note that a relationship between variables similar to~\cref{equation:thermalization_equation} cannot be established when bonds follow the $\{\pm 1\}$ model~\cite{equilibrium_bimodal}. The results can be visualized in~\cref{fig:magnetization_spin_glass}(f), where once again our generalization accelerates the thermalization process.

While the gain obtained by using the generalized Houdayer move might only seem to be marginal, it can vary depending on certain parameters. To gain deeper understanding of how the difference between the original version of the move and our generalization compare, we perform two additional experiments. In the first one, the size of the lattice is fixed to $L\times L$ with $L=100$ and the inverse temperature takes values $\beta\in\{2, 4, 8, 16, 32\}$. At a given temperature, the system is simulated using the original and UCC Houdayer move. For each variant of the move, the number of iterations $n^{\mathrm{variant}}$ needed so that the relative error of the average energy per spin falls below $0.1$ is recorded. We then plot the quantity $n^{\mathrm{Original}}/n^{\mathrm{UCC}}$ as a function of $\beta$ to find the speedup obtained by using the UCC variant of the Houdayer move over the original one. In the second experiment, the same procedure is conducted but the inverse temperature is fixed to $\beta=32$ and the lattice size takes values $L\times L$ with $L\in\{25, 50, 100, 200\}$. The quantity $n^{\mathrm{Original}}/n^{\mathrm{UCC}}$ is then computed as a function of $|V|$. 
\begin{figure*}[!htbp]
    \includegraphics[width=\textwidth]{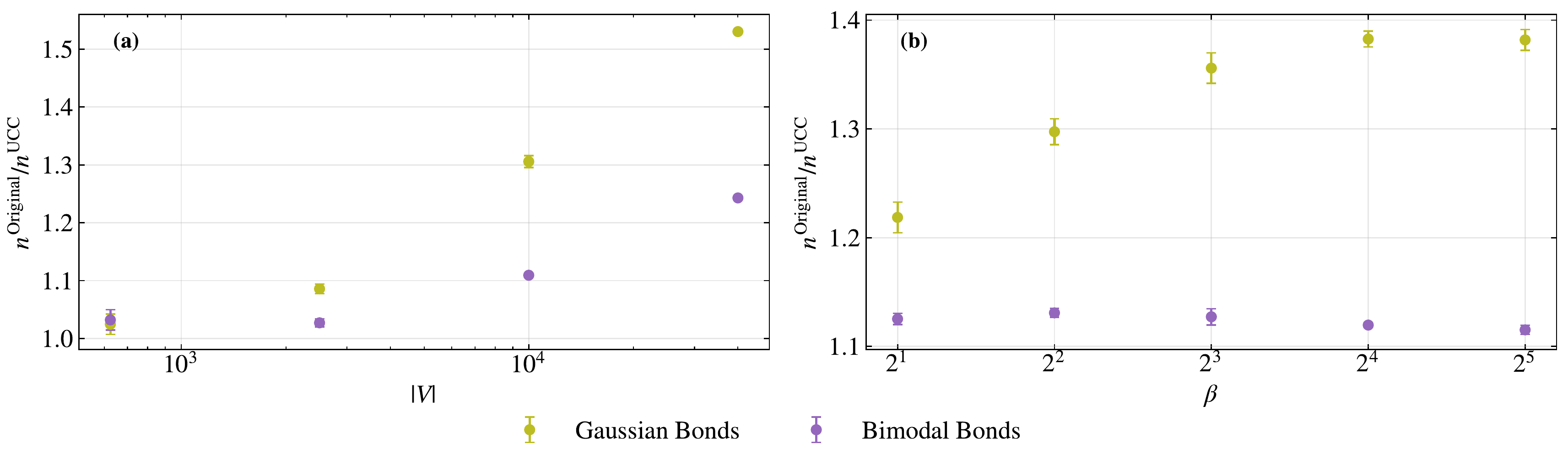}
    \caption[Ratio Performance Houdayer VS UCC]{(a) Error plot showing the evolution of $n^{\mathrm{Original}}/n^{\mathrm{UCC}}$ as a function of $|V|$ when $\beta=32$ (where $n^{\mathrm{variant}}$ is the number of iterations needed for a particular Houdayer variant to reach a relative error of 0.1 for the average energy per spin). Results are shown for Gaussian and bimodal bonds. The average over 10 random bonds assignments is shown in each case, together with the standard error of the mean. (b) Error plot showing the evolution of $n^{\mathrm{Original}}/n^{\mathrm{UCC}}$ as a function of $\beta$ when the lattice size is $100\times 100$ (where $n^{\mathrm{variant}}$ is the number of iterations needed for a particular Houdayer variant to reach a relative error of 0.1 for the average energy per spin). Results are shown for Gaussian and bimodal bonds. In each case, the average over 10 random bonds assignments is shown, together with the standard error of the mean.}\label{fig:spin_glass_improvement}
  \end{figure*}
\Cref{fig:spin_glass_improvement} suggest that the improvement brought by the generalized Houdayer move increases as $|V|$ grows for the two types of bonds considered. As a function of $\beta$, the speedup eventually reaches a plateau. This is expected since when $\beta \to \infty$, the sampling process is essentially transitioning from one ground state to the other, hence remaining at the same energy. At this point, the Houdayer moves are doing most of the work within the MH algorithm, so that the quantity $n^{\mathrm{Original}}/n^{\mathrm{UCC}}$ reflects the best ratio attainable.

Overall, the use of more general variants of the Houdayer move can bring a considerable advantage when sampling from the Boltzmann-Gibbs distribution of both ferromagnetic and spin-glass systems. For small-scale systems, the observed gains can seem modest, but they get increasingly better as the number of variables grows. Comparing the UCC and UHD move suggests that there is almost no difference between the two, indicating that the additional computing power consumed for the UHD move is not necessary. Since the performance of Houdayer moves is sensitive to the graph topology underlying the system, it would be interesting to find a family of problem instances where the use of the UHD Houdayer move brings a larger advantage over the original and UCC versions.
\iffalse
\Cref{fig:delta_evolution} depicts the evolution of $\Delta_{\beta, n}$ as a function of $\log_2(n)$ for $\beta=2$ for two different graphs: a $32\times 32$ square lattice and a $20\times 20$ hexagonal lattice. To obtain the results, the MH algorithm is run for $2^{14}$ iterations for each variant of the Houdayer move, where each iteration consists in $|V|$ single-flips followed by an application of the chosen Houdayer move. This procedure is repeated for 100 different different bond values assignments, and the average behavior is observed together with the standard error of the mean. The results show that the addition of Houdayer moves enables a faster convergence, but the performance between the different versions of the move are similar, with a negligible improvement when using the UCC or UHD variants of the move.

Given the more powerful nature of the generalized Houdayer move, it would be interesting to find a family of problem instances where the use of the UCC and UHD Houdayer move brings a larger advantage over the original version.
\fi
    
    \subsection{Generation of Maximum Cuts} \label{subsection:maxcut_optima}
    In this set of experiments, we explore the performance of \cref{alg:genetic_algorithm} for generating optimal solutions (see \cref{section:ground_state_generation}) in the context of the \textit{Maximum Cut} problem. As a reminder, the aim of the Maximum Cut problem is, given a graph $G=(V, E)$, to find a partition of the nodes into two sets $A$ and $B$ which maximizes the size of the set $\{(i, j)\in E : i\in A \text{ and } j\in B\}$. The corresponding function to minimize can be written as~\cite{max_cut_ising}:
\begin{equation}\label{equation:maxcut_ising_formulation}
    \calH(\config) = \sum_{(i, j)\in E}\sigma_i\sigma_j, 
\end{equation}
where $\config \in \{-1, 1\}^{|V|}$.
Since the goal here is to generate new optimal solutions when given a few of them, we need to consider instances which are highly degenerate. Thus, graphs which present many regularities such as lattices, provide good candidates to run our experimentation. Here, we focus on triangular lattices of sizes $6\times 5$ and $6\times7$. Similar experiments have been completed in~\cite{icm_generating_more_optimum} in order to show the performance that can be obtained by using the Houdayer move. Here, we provide a comparison between different variants of the move, namely the original and the UCC version. The UHD Houdayer move is not used in this case since the genetic algorithm starts with optimal configurations. Indeed, the goal is not about improving the value of the function or the diversity of the pool of configurations given as input, but rather to generate as many new optimal configurations as possible. Hence, it would not make sense to bias our choice towards certain pairs of configurations when performing a Houdayer move, hinting at the fact that one should use the UCC Houdayer move, which selects pairs of configurations in a uniform way (see~\cref{equation:extended_houdayer_distribution}).

For each graph instance considered, the set of all maximum cuts $\X_0$ are computed via a brute-force approach since $|V|$ is not too large. We then give as input to the algorithm subsets of $\X_0$ of particular sizes to study how many more optimal configurations can be computed. The subsets sizes take the form $f\cdot |\X_0|$, where $f \in\{0, 0.05, \dots, 0.95, 1\}$. For each value of $f$, we report the average behavior over 10 uniformly random subsets of the corresponding size. The version of \cref{alg:genetic_algorithm} presented in \cref{section:ground_state_generation} is used, meaning that the mutation and selection steps are removed. Moreover, we restrict ourselves to a single iteration of the algorithm, where the pair-construction step is done by sampling a constant number $K$ of pairs from the pool of all configurations available. We make this number vary as $K\in\{10, 50, 250\}$ in order to analyze how this affect the performance of the algorithm. Finally, the crossover step is done by repeatedly applying the desired flavor of the Houdayer move 100 times on each constructed pair of configurations.

\begin{remark}
Since the function to minimize in \cref{equation:maxcut_ising_formulation} does not contain any linear term, Remark~\ref{remark:houdayer_quadratic_only} applies and one can hence also select clusters with any overlap value in the local overlap. In our experiments, we analyze the difference this can make by repeating the experiments and taking into account these extra clusters (see lines with the $^{++}$ indicated in \cref{fig:maxcut_trilattice}).
\end{remark}
\begin{figure*}[!htbp]
    \includegraphics[width=\textwidth]{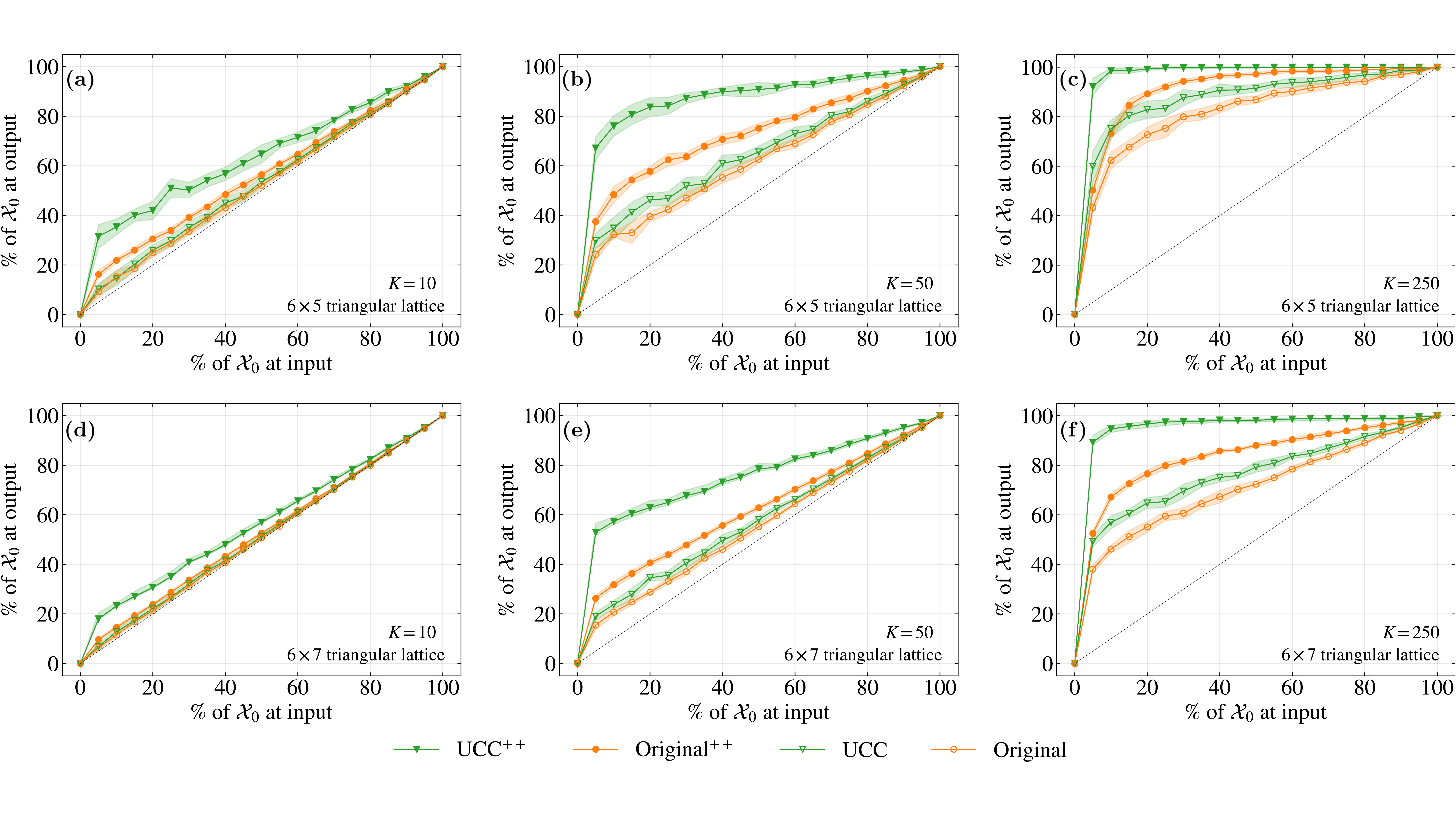}
    \caption[Maximum Cut on Triangular Lattices]{Results for the Maximum Cut example in Section~\ref{subsection:maxcut_optima}, comparing the fraction of all optimal configurations $\X_0$ that can be generated via \cref{alg:genetic_algorithm} given a certain fraction of $\X_0$ as input. The results for a 6 by 5 triangular lattice are shown in (a)--(c), while results for a 6 by 7 triangular lattice are depicted in (d)--(f). Each column contains results for different values of $K$, the number of pairs created from the pool of configurations available. For each plot, both the original and UCC versions of the Houdayer move are compared, together with those making use of \cref{remark:houdayer_quadratic_only} (signaled by a $^{++}$). Each point represents the result obtained for a specific instance averaged over 10 independent runs, and the shaded region around the points show the standard error of the mean.}
    \label{fig:maxcut_trilattice}
  \end{figure*}
\iffalse
\begin{figure*}[ht]
    \begin{minipage}{\textwidth}
    \includegraphics[width=\textwidth]{images/triangular_lattice_6_5_3plots_none2.pdf}
    \vspace{0.01\textwidth}
    \end{minipage}
    \begin{minipage}{\textwidth}
    \includegraphics[width=\textwidth]{images/triangular_lattice_6_7_3plots_none2.pdf}
    \end{minipage}
    \caption[Maximum Cut on Triangular Lattices]{Results for the Maximum Cut example in \cref{subsection:maxcut_optima}. The plots show the fraction of all optimal configurations $\X_0$ that can be generated given a certain fraction of them, using \cref{alg:genetic_algorithm}. Top row shows results for a 6 by 5 triangular lattice and bottom row depicts the results for a 6 by 7 triangular lattice. Each column contains results for different values of $K$, the number of pairs created from the pool of configurations available. In each plot, both the original and UCC versions of the Houdayer move are compared, together with those making use of \cref{remark:houdayer_quadratic_only} (signaled by a $^{++}$). Each point represents the result obtained for a specific instance, averaged over 10 independent runs. The shaded region around points show the standard deviations.}
    \label{fig:maxcut_trilattice}
  \end{figure*}
\fi
 At first glance, the results in~\cref{fig:maxcut_trilattice} show a notable improvement when taking into account the remark above. In every case, considering extra clusters in the local overlap can only offer more possibilities to find optimal configurations that have not been explored yet. It is also evident that the UCC variant of the Houdayer move outperforms the original, which comes to no surprise since the aim was to allow access to a more diverse set of pairs of configurations at the same total function value.

When strongly restricting the number of pairs constructed from the pool of configurations during the crossover phase ($K=10$), the algorithm can only generate a few more of them but it is not enough to find most (see~\cref{fig:maxcut_trilattice}(a),(d)). As $K$ increases, there are more and more opportunities to construct pairs which allow a successful use the Houdayer move, which is exhibited in~\cref{fig:maxcut_trilattice}(b),(e). For $K=250$, the plots reflect that only a small percentage of $\X_0$ is required in order to recover most of it when using the UCC Houdayer move with the comments from \cref{remark:houdayer_quadratic_only}. More precisely, only having access to 5\% of $\X_0$ allows to retrieve around 90\% of it, while the knowledge of 10\% of $\X_0$ is sufficient to retrieve almost all optimal configurations. In comparison, the original Houdayer move recovers only 75\% of all optimal configurations or less when given 10\% of them (see~\cref{fig:maxcut_trilattice}(c),(f)).
    
    \subsection{Generation of Vertex Covers}\label{subsection:minimum_vertex_cover}
    The \textit{Minimum Vertex Cover} (MVC) problem is a well-known constrained optimization problem which can be tackled by \cref{alg:genetic_algorithm}. Given a graph $G=(V, E)$, the goal is to find a set of vertices $S \subseteq V$ of minimum cardinality with the constraint that it ``covers'' all edges in $G$. For a given $S$, an edge $(i, j) \in E$ is said to be covered if at least one of its constituent vertices is part of $S$.  The corresponding function to minimize can be derived  as~\cite[Section 4.3]{ising_formulation_np_problems}:
\begin{equation*}
    \mathcal{H}(\boldsymbol{x}) = \sum_{i \in V}x_i + P\sum_{(i, j) \in E}\left(1-x_i\right)\left(1-x_j\right),
\end{equation*}
where $\boldsymbol{x}\in \{0, 1\}^{|V|}$ and $P>1$ is a parameter controlling the strength of penalization of configurations which are not covers. This can be conveniently converted to the form of \cref{equation:problem_statement} via \cref{remark:ising_one_to_one_mapping}.

For these experiments, we restrict our attention to the MVC problem on random $d$-regular graphs, which is known to be NP-complete~\cite[Theorem 1]{vertex_cover_np_regular_graphs}. Furthermore, this family of graphs is interesting because their percolation threshold $p_c$ can be analytically determined and has value $p_c=\frac{1}{d-1}$~\cite[Section 3]{percolation_threshold_expanders}. Focusing on the case $d=3$ ensures that $p_c \geq \frac12$, giving room for Houdayer moves to perform successfully as explained in \cref{section:houdayer_limitations}.

Starting from a set of configurations $S_{\mathrm{init}}$ given by the simulated annealing algorithm with single-flip proposals, we use \cref{alg:genetic_algorithm} to build a set $S_{\mathrm{final}}$ of better configurations. Two metrics are analyzed: the improvement over the initial population in terms of the value of $\mathcal{H}$ attained, and the improvement in terms of the diversity of the configurations generated by the genetic algorithm. Note that similar types of experiments have been performed in~\cite[Section II.A]{icm_generating_more_optimum} using the original form of the Houdayer move. We show here how other versions of the move affect the performance of the algorithm.

The improvement in terms of the function value is captured in the following way. We use $\mathcal{H}_0^{\mathrm{init}}$ and $\mathcal{H}_0^{\mathrm{final}}$ to denote the lowest function values attained by the set $S_{\mathrm{final}}$ and $S_{\mathrm{init}}$, respectively. Given that $\mathcal{H}(\boldsymbol{x}) \geq 0$ and since the genetic algorithm can only output a solution with value at least as good as $\mathcal{H}_0^{\mathrm{init}}$, the quantity $1-(\mathcal{H}_0^{\mathrm{final}}/\mathcal{H}_0^{\mathrm{init}}) = (\mathcal{H}_0^{\mathrm{init}} -\mathcal{H}_0^{\mathrm{final}})/\mathcal{H}_0^{\mathrm{init}}$ lies in the interval $[0,1]$. This provides an easy way to measure the improvement achieved by the genetic algorithm: for example, the value $(\mathcal{H}_0^{\mathrm{init}} -\mathcal{H}_0^{\mathrm{final}})/\mathcal{H}_0^{\mathrm{init}} = 0.5$ means that the genetic algorithm found a solution with function-value equal to half of the best one given as input, thus yielding a 50\% improvement. We emphasize that although simulated annealing was used here to compute $S_{\mathrm{init}}$, the initial set can in general be chosen via any means.
\begin{figure*}[!htbp]
    \centering
    \includegraphics[width=\textwidth]{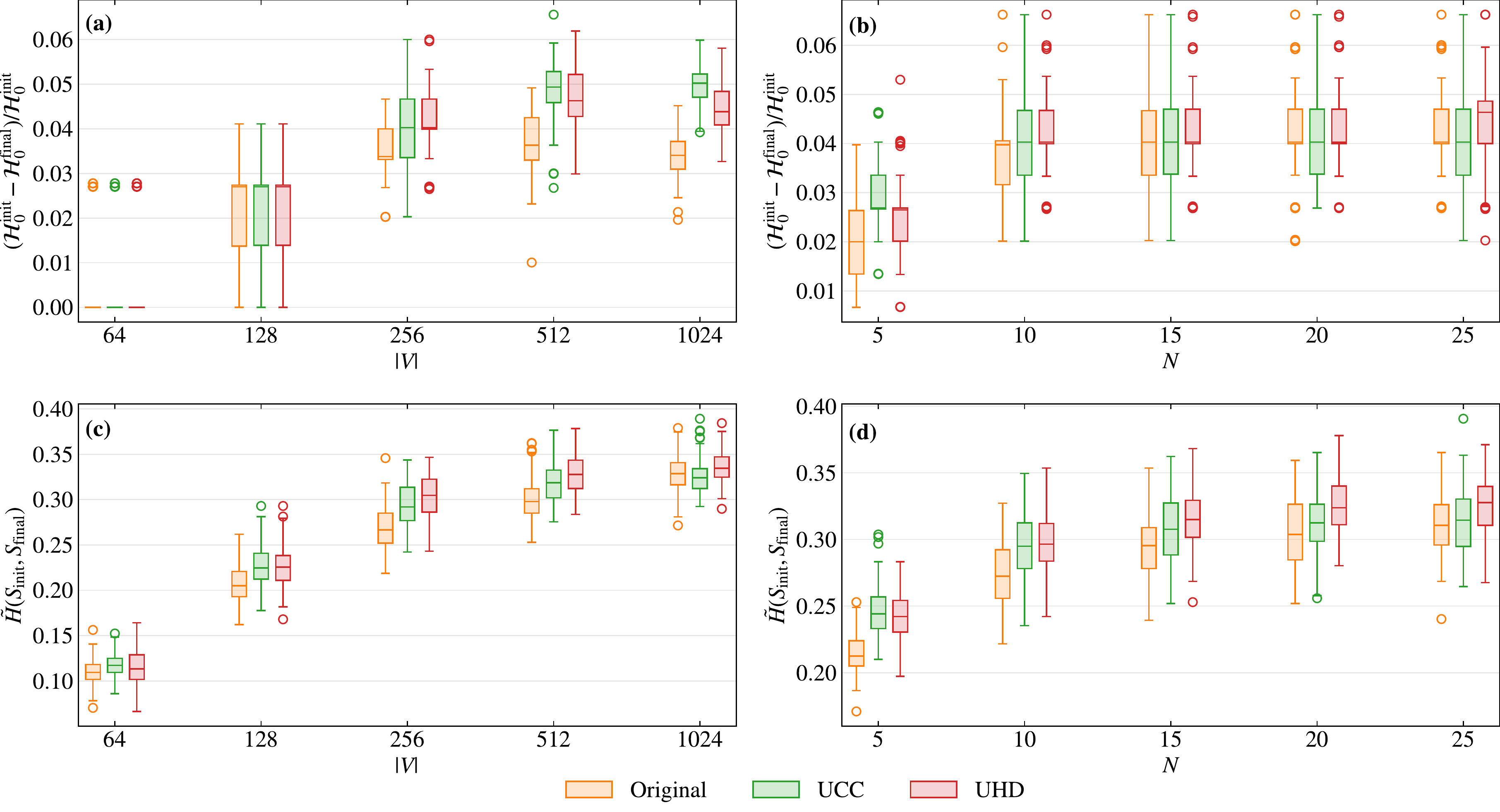}
    \caption{Boxplots of the results for the example in \cref{subsection:minimum_vertex_cover} comparing the performance of the genetic algorithm when using three different versions of the Houdayer move: the original one, the UCC one and the UHD one. The plots (a)--(b) show the values obtained for the ratio $(\mathcal{H}_0^{\mathrm{init}} -\mathcal{H}_0^{\mathrm{final}})/\mathcal{H}_0^{\mathrm{init}}$, while the plots (c)--(d) focus on the modified Hausdorff distance $\Tilde{H}(S_{\mathrm{init}}, S_{\mathrm{final}})$ between the set of configurations $S_{\mathrm{init}}$ given as input to the genetic algorithm and the resulting set $S_{\mathrm{final}}$. The plots in (a) and (c) study these measures as a function of the system size $|V|$ with fixed number of iterations $N=10$, while the plots in (b) and (d) do so as a function of the number of iterations $N$ performed by the genetic algorithm with fixed system size $|V|=256$.}
    \label{fig:min_vertex_cover_results}
\end{figure*}

To quantify the diversity of the solutions computed, a measure of how far the initial pool is from the final pool is used, through a notion of distance between sets inspired by the \textit{Hausdorff distance}~\cite{hausdorff_distance}. Originally defined as
\begin{equation}\label{equation:hausdorff_distance}
    H(S, T) = \max\left(\max_{\boldsymbol{\tau}\in T} d_H(\boldsymbol{\tau}, S), \max_{\config\in S} d_H(\config, T)\right),
\end{equation}
where $S, T \subseteq \mathcal{X}$, the Hausdorff distance has the inconvenience to be sensitive to outliers due to the $\max(\cdot)$ operator. To add robustness against outliers, we use one of its modified versions proposed in~\cite{modified_hausdorff_distance} and replace $\max(\cdot)$ operations by taking medians instead. At the same time, a normalization by $|V|$ is also included to ease the comparison of such distances between systems of different sizes. Overall, this yields the modified quantity
\begin{align}\label{equation:hausdorff_distance}
    \Tilde{H}(S, T) = \frac{1}{2|V|}\Bigl(\mathrm{med}&(\{d_H(\boldsymbol{\tau}, S): \boldsymbol{\tau}\in T\})\nonumber\\
    &+\mathrm{med}(\{d_H(\config, T): \config\in S\})\Bigr),
\end{align}
which lies in the interval $[0, 1]$.

We study the two metrics presented above as a function of the system size $|V|$ and the number of iterations $N$ performed by the genetic algorithm. First, we consider systems of increasing sizes with $|V|\in\{64, 128, 256, 512,1024\}$. For each size considered, 100 random instances are generated and the genetic algorithm is applied to each of them while keeping track of the metrics of interest. A version of the genetic algorithm which runs for $N=10$ iteration is used, where at each step 500 random pairs are chosen uniformly from the pool of configurations. The crossover step is done using different strategies for picking Houdayer clusters, namely via the original, UCC, and UHD Houdayer moves. The selected type of move is applied 10 times on each created pair to ensure that many pairs are generated whenever a Houdayer move can be successfully performed. The mutation operator flips each variable with probability $1/|V|$, and the selection process regulates the pool size by retaining 250 configurations. Similar experiments are then repeated, but this time the number of iterations is increased as $N\in\{5, 10, 15, 20, 25\}$ while the system size is fixed to $|V|=256$. The results are gathered and depicted in \cref{fig:min_vertex_cover_results}.
A quick glance at~\cref{fig:min_vertex_cover_results}(a),(c) reveals that the use of the UCC and UHD Houdayer moves yields better results with respect to the two measures considered. Indeed, since those moves are specifically designed to consider all combinations of Houdayer clusters, they have the ability to find a much richer pool of configurations both in terms of the function $\mathcal{H}$ and the diversity measure $\Tilde{H}$. There are however still nuances, as shown in the top-left plot which suggests that the three methods perform similarly for systems of small sizes ($|V|<256$), and in fact poorly when $|V|=64$. As the size of the system grows, the UCC Houdayer move gives the best results in terms of the function value, followed by the UHD Houdayer move. The bottom-left plot comparing the modified Hausdorff distance shows a distinct pattern: the UHD variant performs the best, followed by the UCC one, itself followed by the original move. Intuitively, by offering access to many more pairs via the UCC move, this allows the Hamming distance between configurations to vary extensively, resulting in a larger value of $\Tilde{H}$. By further tailoring the move to make it transition in a uniform way with respect to the Hamming distance, this further improves the diversity of configurations found.

\Cref{fig:min_vertex_cover_results}(b),(d) depicts how the number of iterations of the genetic algorithm affects performance. The results confirm what the intuition may suggest: as the number of iterations grow, more mixing happens between configurations, which reduces the advantage offered by the UCC and UHD Houdayer moves. In the top plot, all variants give similar results after 15 iterations or more, even though the UHD variant has a slight advantage. For the modified Hausdorff distance, the same comments as previous paragraph can be made.

To sum up, when the goal is to generate as many distinct configurations as possible, one has to look into all possible combinations of Houdayer clusters. Specifically, \cref{fig:min_vertex_cover_results} indicates that better results in terms of the function can be attained via the UCC Houdayer move while a larger diversity of configurations is obtained when using the UHD Houdayer move. Using such techniques enables to improve on the best function value attained by the initial pool by $3$-$6\%$. In terms of the diversity of the configurations generated, our modified Hausdorff distance shows that roughly, the configurations at the input differ with those at the output on $30$-$40\%$ of all variables.

    %\subsection{Image Denoising using the Ising Model} \label{subsection:image_denoising}
    %\input{numerical_experiments/image_denoising}

\section{Conclusion and Outlook} \label{section:conclusion}
    Solving general quadratic unconstrained binary optimization problems, or equivalently Ising spin-glass instances, is a hard computational task. In this work, we explored the use of the Houdayer cluster algorithm and its ability to improve both sampling and optimization. This examination led to an extension of the move (cf.~\cref{section:generalized_houdayer}), giving more possible ways to find pairs of configurations at the same energy. This also enabled us to understand the main purpose and strength of the Houdayer move: by allowing to mix configurations together and consequently making large rearrangements in the variables' values, it offers a unique strategy to explore the space of all configurations. In particular, the new variants can be used both to improve sampling in Markov chain Monte-Carlo procedures (where one can adapt it to yield a better exploration of the space, cf.~\cref{subsection:improved_sampling}), but also as a building block in a genetic algorithm which can enhance the quality of a given pool of configurations (cf.~\cref{subsection:genetic_algorithm}). The experiments suggest that the generalization leads to improvements over the original move in a variety of examples (cf.~\cref{section:numerical_results}), even though the difference between the two might not always be significantly large.

Multiple directions can be further explored to extend the results presented in this work. One of the most useful additions would be to devise a method to make Houdayer moves successful on a larger family of graphs. As mentioned in~\cref{section:houdayer_limitations}, the move is likely to fail on graphs with percolation threshold lower than 0.5. Could one come up with a method to modify the graph structure in order to retain an equivalent problem statement while reducing the percolation threshold of the graph? Alternatively, if the preservation of energy is not a strong requirement, one could make the creation of clusters similar to the Swendsen-Wang-Wolff algorithm~\cite{cluster_move_diluted_spins, percolation_signature_spin_glass}. In that case, more clusters can be created and their sizes could be further tuned by allowing the use of arbitrary probabilities when adding sites to clusters~\cite{general_cluster_updating_method}.

Another direction to examine would be to design a modification of the Houdayer move for the Potts model~\cite{Potts1952}. In the $q$-Potts model, variables can take one of $q\geq 2$ values and the function of interest is of the form
    \begin{equation*}
        \mathcal{H}(\config) = -\sum_{(i, j)\in E} J_{ij} \delta_{\sigma_i\sigma_j} - \sum_{i\in V} h_i \sigma_i,
    \end{equation*}
    where $\boldsymbol{\sigma}\in \{0, \dots, q-1\}^{|V|}$ and $J_{ij}, h_i\in \mathbb{R}$. In this setting, a procedure similar to the Houdayer move can be designed by maintaining $q$ configurations at the same time instead of just a pair. At each step, the move would compute the local overlap between all $q$ configurations, thus determining the sites where all configurations disagree. Just as in the generalized procedure presented, these sites would define clusters, from which a combination of them can be chosen. One can then take the values of the variables in the $q$ configurations at these sites and exchange them in an arbitrary manner between configurations while preserving the energy. Assuming that the $q$ configurations are given uniformly at random, the probability that all $q$ configurations disagree at a particular site is equal to $\frac{q!}{q^q}$. Much like the case where $q=2$, the process of cluster formation can then be viewed as a site percolation process with occupation probability $p(q)=\frac{q!}{q^q}$. In this case, the formation of a spanning cluster would result in a failure if the site percolation threshold $p_c$ of the graph studied is such that $p_c < p(q)$. Since $p(q)$ decreases exponentially fast with respect to $q$, for $q\geq 2$ it becomes highly unlikely for the move to fail due to a low graph percolation threshold. However, the price to pay comes from the fact that the exponential decrease of $p(q)$ corresponds to an exponential decrease in the number of sites where all configurations disagree. Ultimately, this results in the nonexistence of any Houdayer cluster unless the number of nodes grows as $\Omega(1/p(q))$. This generalization is thus interesting mainly for low values of $q$ such as $q=3$, where the percolation threshold $p(q)$ takes lower values which allow the move to work for more graph topologies, while remaining large enough to ensure the existence of Houdayer clusters in the local overlap. Interesting graph problems like finding a $q$-coloring~\cite[Section 8.7]{algorithms_design} could be tackled via this approach and compared to conventional methods which rely solely on the Metropolis-Hastings algorithm with local updates.
    
    The purpose of the Houdayer move is to mix pairs of configurations together while preserving their total function value. A natural question to ask is whether one can mix more than just two configurations at the same time. We touched upon this in the previous paragraph when considering the general $q$-Potts model, but these ideas could already be used in the usual Ising model.
    
    Further experiments could be conducted to understand how well our generalization performs against other methods in various settings. There could be for example interesting applications in the field of restricted Boltzmann machines~\cite{rbm_hinton} where parallel tempering has been shown to improve their training~\cite{pt_efficient_rbm, tempered_mcmc_rbm, pt_rbm_convergence}. By additionally applying Houdayer moves, the training could be further improved, in particular for convolutional restricted Boltzmann machines~\cite{convolutional_rbm}.
    
    Finally, a more theoretical view on the Houdayer move could be taken in the context of Markov chain Monte-Carlo methods in order to understand the mixing time of the constructed Markov chains. As highlighted in \cref{section:expected_number_of_flips}, a direct comparison with the simple single-flip dynamics might not be possible, but the study of transition matrices could lead to a more profound understanding of the Houdayer move's apparent ability to reduce the mixing time.
    
Overall, we believe that a better understanding of the mathematical underpinnings of Houdayer moves will open the doors to new improvements in the simulation of spin systems, and more generally in the domains of sampling and optimization. Our generalization and its applications are only a small step towards a full comprehension of what this algorithm can provide. Being at the intersection of multiple fields like percolation theory or Markov chain theory, there is no doubt that many angles can be taken to make advances towards a rigorous analysis. Other paths could lead to progress, in particular by tackling the current limitations highlighted throughout the document. It is our hope that this work will serve as a guide for such further developments and as a source of inspiration towards a deeper comprehension of spin-glass systems.
    
\section*{Acknowledgments}
The authors would like to extend their gratitude to Amedeo Roberto Esposito and Ulrike Bornheimer for thoroughly reviewing the manuscript. We also thank Nicolas Macris and Alistair Sinclair for helpful discussions.

\section*{Appendix}

\appendix
 
\section{Percolation Theory and Houdayer Moves}\label{appendix:percolation_theory}

This section is dedicated to a brief introduction to percolation theory and its relationship to Houdayer moves, with the aim of providing the reader with the necessary notation and terms used in this field. For a more in-depth introduction, we invite the reader to look at~\cite{Beffara2006, intro_percolation_theory}.

The model considered is the following: given a graph $G=(V, E)$ and $p\in [0, 1]$,  we look at the probabilistic process where each vertex is deemed ``occupied" with probability $p$ and ``empty" with probability $1-p$. The main question of interest is whether for a given $p$ there exists a cluster of nodes which spans the entire graph. That model is called \textit{site percolation}, not to be confused with its counterpart called \textit{bond percolation}, in which the edges of a graph are removed with a certain probability and one is interested in studying the subgraph resulting from this process. We note that most of the literature focuses on graphs with infinite number of nodes. 

For an arbitrary node $i\in V$, we denote by $S_i(p)$ the random variable representing the size of the cluster containing $i$ when using occupation probability $p$. One central quantity in percolation theory is the percolation threshold of a graph $G$, defined as
\begin{equation*}
    p_c(G) \triangleq \sup\left\{p\in [0, 1]: \forall i\in V,\; \mathbb{P}(S_i(p) = \infty) = 0\right\}.
\end{equation*}

This is the highest value the occupation probability can take which ensures that no infinite cluster spans the graph. It turns out that this model exhibits a phase transition at $p_c$, meaning that the behavior of the system is very different depending on whether $p<p_c$ or $p>p_c$. In particular in the regime $p<p_c$, the probability of finding an infinite cluster is close to 0, while in the regime $p>p_c$ it is close to 1. For finite-sized graphs there is in general no sharply defined threshold, but a pseudo phase transition happens~\cite{cluster_algo_any_space_dimension}. In that case, denoting by $S_{max}(p) \triangleq \max_{i\in V} S_i(p)$ the size of the largest cluster, we say that $p_c$ is the probability threshold for which $p<p_c$ results in $\frac{\mathbb{E}[S_{max}(p)]}{|V|} \to 0$, while $p>p_c$ yields $\frac{\mathbb{E}[S_{max}(p)]}{|V|} \to 1$.

With these definitions, the link between site percolation and Houdayer moves is straightforward. Consider some function $\mathcal{H}: \X\to\mathbb{R}$ of the form of \cref{equation:problem_statement} with underlying graph $G$. Given two configurations $\config^{(1)}, \config^{(2)}\in \X$, recall that the computation of the local overlap $q_i$ when performing a Houdayer move, cf.~\cref{section:description_houdayer}, amounts to determining where variables agree and disagree between the two configurations. Assuming that each configuration is picked uniformly at random in $\X$ and focusing on a particular site $i\in V$, we thus have
\begin{equation*}
    \mathbb{P}(q_i=-1) = \mathbb{P}(\sigma_i^{(1)}\neq\sigma_i^{(2)}) = \frac12.
\end{equation*}

Hence the formation of Houdayer clusters, which are nothing but groups of connected sites with $q_i=-1$, is explained by a site percolation process with occupation probability $\mathbb{P}(q_i=-1)=\frac12$.

\section{A Houdayer Move Tailored for Markov Chain Monte-Carlo Sampling}\label{appendix:optimal_mcmc_houdayer}
Houdayer moves find applications in various contexts, one of the most well-known being Markov Chain Monte-Carlo sampling. As explained in Section \ref{subsection:improved_sampling}, in order to design a Monte-Carlo move which yields a good exploration of the space of configurations, one should be able to transition to configurations which are at all kinds of Hamming distances, ideally in a uniform way. In particular for the Houdayer move, the space to be explored is the space of all pairs of configurations of variables.

To develop such a move, the number of flips induced by each combination of Houdayer clusters needs to be determined, so that one can choose combinations which lead to pairs of configurations uniformly according to the Hamming distance. This is achieved by encoding the information about the number of flips of each combination in a table, and by then traversing this structure to retrieve a combination of Houdayer clusters which performs a desired number of flips. Interestingly, notice that this problem is a special case of the \textit{subset-sum problem}: one is given a set of numbers (in our case the number of flips incurred by each Houdayer cluster), and seeks to determine the sum of subsets of those numbers (which in our case corresponds to the number of flips incurred by combinations of clusters). This type of problem can be solved in pseudo-polynomial time using dynamic programming~\cite[Section 6.4]{algorithms_design}.

To ease the notation, given a pair of configurations $p \in \XX$, we denote the corresponding set of Houdayer clusters as $\C(p)=\{C_1, C_2, \dots, C_k\}$ where $|\C(p)|=k$, and the number of variables covered by all of them as $l = \sum_{i=1}^k |C_i|$. The move then works as follows:
\begin{enumerate}
    \item Create and fill a table $A$ of size $(k+1) \times (l+1)$ via dynamic programming, where the entry in row $i \in \{1, \dots, k+1\}$ and column $j \in \{1, \dots, l+1\}$ contains the number of combinations $T \subseteq \{C_1, \dots, C_{i-1}\}$ such that $\sum_{C\in T}|C| = j-1$. Thus, each element in the last row $A_{k+1, j}$ with $j\in \{1, \dots, l+1\}$ indicates the number of combinations $T\subseteq \C(p)$ that flip $j$ variables. Note that when $i=1$, we use the convention $\{C_1, \dots, \C_{i-1}\}=\emptyset$.
    \item  Choose $j$ uniformly at random in $\{1, \dots, l+1\}$ and reconstruct one of the combinations achieving $(j-1)$ number of flips uniformly at random by traversing the table. Once the combination of Houdayer clusters is chosen, the rest of the move can be carried out as usual by flipping the values of the variables covered by the selected combination.
\end{enumerate}
More precisely, the first step is done as described in~\cite[Section 6.4]{algorithms_design} using a bottom-up approach to fill the table in $O(kl)$ time~\cite[Equation (6.9)]{algorithms_design}. To retrieve a solution given some column $j$ in the second step, one can adapt the procedure described in~\cite[Equation (6.5)]{algorithms_design} in order to recover uniformly at random one of the solutions which achieves $(j-1)$ flips, taking $O(k)$ time~\cite[Equation (6.10)]{algorithms_design}. A simple example is depicted in Fig. \ref{fig:table_optimal_mcmc_move}.
\begin{figure}
    \centering
    \includegraphics[width=\columnwidth]{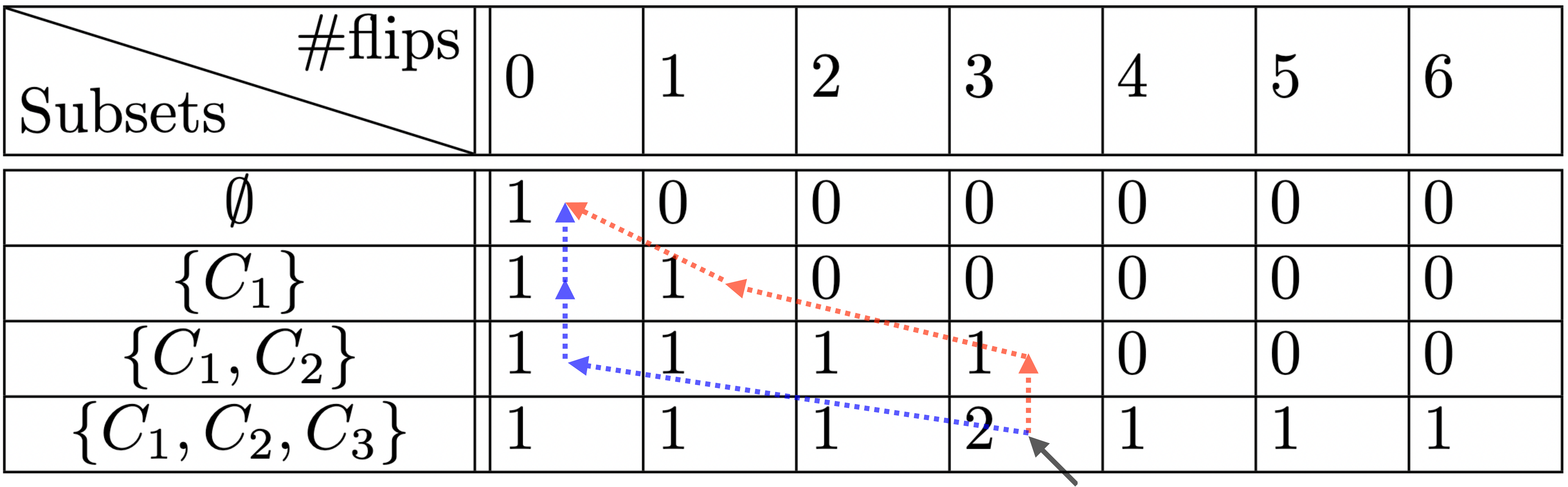}
    \caption{Example of the table resulting from dynamic programming when $\C =\{C_1, C_2, C_3\}$ with $|C_i|=i$. The table encodes the number of flips achievable by all possible combinations. Here, the arrows show how the table is used to trace back the two possible combinations leading to flipping 3 variables. In each path, an arrow going up indicates that a cluster is not chosen, while an arrow going left indicates that a cluster is selected. The red path corresponds to the combination $\{C_1, C_2\}$ and the blue path to combination $\{C_3\}$.}
    \label{fig:table_optimal_mcmc_move}
\end{figure}
\begin{remark}
In practice, selecting $j\in \{1, \dots, l+1\}$ in the second part is prone to result in a failure since it is often the case that no combination achieves $(j-1)$ flips. To prevent such a behavior, one can restrict the choice to be in the set $J\triangleq \{j \in \{1, \dots, l+1\}: A_{k+1, j} > 0\}$, effectively considering only those values which have at least one corresponding combination of Houdayer clusters. Note that even though this forces the move to always return a valid combination, choosing uniformly in $J$ might however break the uniformity with respect to the Hamming distance. This effect can be mitigated by adjusting the way we select from $J$, by for example using a maximum entropy distribution~\cite{max_entropy} which ensures that the mean number of flips is the same as if the choice was uniform in $\{1, \dots, l+1\}$.
\end{remark}
\begin{remark}
Notice that the Hamming distance between pairs can be changed to be the ``true" number of flips $n_{\C}$ as defined in \cref{equation:true_number_flips}. That is, we know from \cref{lemma:extended_houdayer_properties} that $T$ is in fact equivalent to $T^c$, meaning that one can actually cut the table $A_{i, j}$ in half by only filling for values $j\in \{0, \dots, \lfloor l/2\rfloor\}$. Indeed, all combinations achieving the number of flips in the second half of the table are just the complement combinations of those achieving the first half.
\end{remark}
In this case, it is possible to write $\nu_{\C}$ explicitly. Indeed, given the set of Houdayer clusters, one can partition the set of all combinations $2^{\C}$ into sets $\mathcal{S}_i\subseteq \C$ for $i\in\{0, \dots, l\}$, where $\mathcal{S}_i = \{T\subseteq\C : \sum_{C\in T}|C| = i\}$ denotes the set of combinations which cover exactly $i$ variables. With this, selecting a particular combination $T\subseteq \C$ (which covers $t\triangleq\sum_{C\in T}|C|$ variables) requires to first select $j=t+1$ in step 2 when choosing uniformly in $\{1, \dots, l+1\}$, and then select combination the $T$ among those in $\mathcal{S}_t$. Since each choice is taken uniformly at random, we have that
\begin{equation*}
    \nu_{\C}(T) = \frac{1}{(l+1)\cdot|\mathcal{S}_{t}|}.
\end{equation*}
To understand how much of an improvement it can give, this version of the move is compared to other ones in \cref{section:numerical_results}.

\section{Proof of \cref{lemma:extended_houdayer_same_energy}} \label{appendix:proof_same_energy_lemma}

Let us re-state the result for ease of reference.

\begin{lemma2}
Consider a pair of configurations of variables $(\boldsymbol{\sigma}^{(1)}, \boldsymbol{\sigma}^{(2)})$ and suppose that a generalized Houdayer move results in the new pair $(\boldsymbol{\tau}^{(1)}, \boldsymbol{\tau}^{(2)})$. Then we have $\mathcal{H}(\boldsymbol{\sigma}^{(1)}) + \mathcal{H}(\boldsymbol{\sigma}^{(2)}) = \mathcal{H}(\boldsymbol{\tau}^{(1)}) + \mathcal{H}(\boldsymbol{\tau}^{(2)})$.
\end{lemma2}

\begin{proof}
We give a proof for the general case where variables take values in $\{a, b\}$ for some $a, b, \in \mathbb{R}$. Given some variable $\sigma$, we denote by $\bar{\sigma}$ the operation of ``inverting" its value, i.e., $\bar{\sigma} = \begin{cases}
       b &\quad\text{if } \sigma = a,\\
       a &\quad\text{if } \sigma = b.\\
     \end{cases}$
To avoid lengthy equations, we split the function into linear and quadratic terms as
\begin{equation*}
    \mathcal{H}_{l}(\boldsymbol{\sigma}) = - \sum_{i \in V}  h_i \sigma_i
\end{equation*} 
and 
\begin{equation*}
    \mathcal{H}_{q}(\boldsymbol{\sigma}) = -\sum_{(i, j) \in E} J_{ij}\sigma_i \sigma_j,
\end{equation*} 

so that $\mathcal{H}(\boldsymbol{\sigma}) = \mathcal{H}_{q}(\boldsymbol{\sigma}) + \mathcal{H}_{l}(\boldsymbol{\sigma})$. To show that both the linear and quadratic part of the energy function of the input and output pair remain identical, we mainly use the fact that for a site $i\in V$ chosen to be flipped, it must be that $\sigma^{(1)}_i \neq \sigma^{(2)}_i$ since their overlap is assumed to be $q_i=-1$. 

Suppose that a combination of Houdayer clusters $T \subseteq \mathcal{C}(\boldsymbol{\sigma}^{(1)}, \boldsymbol{\sigma}^{(2)})$ is selected to be flipped according to the desired distribution $\nu_{\C(\boldsymbol{\sigma}^{(1)}, \boldsymbol{\sigma}^{(2)})}$. For the linear part, we have
    \begin{align}
        \mathcal{H}_l(&\boldsymbol{\sigma}^{(1)}) + \mathcal{H}_l(\boldsymbol{\sigma}^{(2)}) -  \mathcal{H}_l(\boldsymbol{\tau}^{(1)}) - \mathcal{H}_l(\boldsymbol{\tau}^{(2)}) \nonumber\\
        &= - \sum_{i \in V}  h_i (\sigma^{(1)}_i + \sigma^{(2)}_i - \tau^{(1)}_i - \tau^{(2)}_i) \nonumber\\
        &= -\sum_{C \in T}\sum_{i \in C}  h_i (\sigma^{(1)}_i + \sigma^{(2)}_i - \tau^{(1)}_i - \tau^{(2)}_i) \label{equation:energy_conserved_proof_linear_part} \\
        &= -\sum_{C \in T}\sum_{i \in C}  h_i (\sigma^{(1)}_i + \sigma^{(2)}_i - \bar{\sigma}^{(1)}_i - \bar{\sigma}^{(2)}_i) \nonumber\\
        &= -\sum_{C \in T}\sum_{i \in C}  h_i (\sigma^{(1)}_i + \sigma^{(2)}_i - \sigma^{(2)}_i - \sigma^{(1)}_i) \nonumber\\
        &= 0,\nonumber
    \end{align}
    where in \cref{equation:energy_conserved_proof_linear_part} we use the fact that for sites $v\in V$ which are not in the chosen combination of clusters, the variables remain unchanged.
    
For the quadratic part, note that terms $J_{ij}\sigma_i \sigma_j$ in the sum are altered only for edges $(i, j)\in E$ such that exactly one the constituent node is covered by $T$. Indeed, if none or both nodes are covered by $T$, the overall contribution remains unchanged. By denoting the set of edges inducing a change of energy, namely those at the boundary of $T$, as $\partial T=\{(i, j)\in E: i\text{ is covered by $T$ and $j$ is not covered by $T$}\}$, we have
    \begin{align*}
        \mathcal{H}_q(\boldsymbol{\tau}^{(1)}) +\mathcal{H}_q(\boldsymbol{\tau}^{(2)})- \mathcal{H}_q(\boldsymbol{\sigma}^{(1)}) - \mathcal{H}_q(\boldsymbol{\sigma}^{(2)}) =&\\
        \sum_{(i, j) \in E} J_{ij} \left(\sigma^{(1)}_i \sigma^{(1)}_j + \sigma^{(2)}_i \sigma^{(2)}_j - \tau^{(1)}_i \tau^{(1)}_j  -\tau^{(2)}_i \tau^{(2)}_j \right) =& \\
        \sum_{(i, j) \in \partial T} J_{ij} \left(\sigma^{(1)}_i \sigma^{(1)}_j + \sigma^{(2)}_i \sigma^{(2)}_j - \tau^{(1)}_i \tau^{(1)}_j - \tau^{(2)}_i \tau^{(2)}_j \right)=&\\
        \sum_{(i, j) \in \partial T} J_{ij} \left(\sigma^{(1)}_i \sigma^{(1)}_j + \sigma^{(2)}_i \sigma^{(2)}_j - \bar{\sigma}^{(1)}_i \sigma^{(1)}_j - \bar{\sigma}^{(2)}_i \sigma^{(2)}_j \right) =&\\
        \sum_{(i, j) \in \partial T} J_{ij} \left(\sigma^{(1)}_i \sigma^{(1)}_j + \sigma^{(2)}_i \sigma^{(2)}_j - \sigma^{(2)}_i \sigma^{(1)}_j - \sigma^{(1)}_i \sigma^{(2)}_j \right) =&\\
        \sum_{(i, j) \in \partial T} J_{ij} \left(\sigma^{(1)}_i - \sigma^{(2)}_i\right) \left(\sigma^{(1)}_j - \sigma^{(2)}_j\right)=&\\
        0=,
    \end{align*}
    where in the last equality, we use the fact that $\sigma^{(1)}_j = \sigma^{(2)}_j$. Indeed it must be that the site overlap at $j$ is equal to $q_j=1$ or otherwise it would be part of Houdayer cluster containing $i$ since $(i, j)\in E$. We can thus conclude that
    \begin{align*}
        \mathcal{H}_q(\boldsymbol{\sigma}^{(1)}) + \mathcal{H}_q(\boldsymbol{\sigma}^{(2)}) =  \mathcal{H}_q(\boldsymbol{\tau}^{(1)}) + \mathcal{H}_q(\boldsymbol{\tau}^{(2)}) .
    \end{align*}
Combining both parts of the proof yields the result.
\end{proof}

\section{Proof of \cref{lemma:extended_houdayer_properties}} \label{appendix:proof_extended_houdayer_properties}

For reference, we re-state the result to be proven.

\begin{lemma2}
Consider a pair of configurations of variables $\left(\boldsymbol{\sigma}^{(1)}, \boldsymbol{\sigma}^{(2)}\right)$ such that their local overlap induces the set of Houdayer clusters $\mathcal{C}$. Then, we have that:
\begin{enumerate}
    \item There are $2^{|\mathcal{C}|}$ unique pairs of configurations (including the initial one) that can be reached.
    
    \item For any pair $\left(\boldsymbol{\tau}^{(1)}, \boldsymbol{\tau}^{(2)}\right)$ that is reachable via some combination of Houdayer clusters $T\subseteq \mathcal{C}$, the pair $\left(\boldsymbol{\tau}^{(2)}, \boldsymbol{\tau}^{(1)}\right)$ is also reachable, using the ``complement" combination $T^c \triangleq \mathcal{C} \setminus T \subseteq \mathcal{C}$.
\end{enumerate}
\end{lemma2}

\begin{proof}
In this case, one can choose any combination of clusters out of the $|\mathcal{C}|$ available ones, and each combination gives rise to a new unique pair. Considering the empty combination (which yields back the input pair) to be valid, the number of possible combinations is given by

\begin{equation*}
    \sum_{i=0}^{|\mathcal{C}|} \binom{|\mathcal{C}|}{i} = 2^{|\mathcal{C}|},
\end{equation*}

which is also the number of unique pairs that can be reached.

To prove the second part of the lemma, consider a combination of clusters $T \subseteq \mathcal{C}$ to be flipped. Since during a Houdayer move, the only sites at which variables can see their values change is covered by $\mathcal{C}$, we restrict our attention to them. To simplify the notation, we write $\boldsymbol{\sigma}_T$ to index variables of some configuration $\boldsymbol{\sigma}$ belonging to sites covered by the combination of clusters $T\subseteq \mathcal{C}$.

Computing the elements of the first configuration $\boldsymbol{\tau}^{(1)}$ in the new pair after swapping, we have that variables covered by $T$ now take values $\boldsymbol{\sigma}_T^{(2)}$, while variables covered by $ \mathcal{C} \setminus T$ remain unchanged, equal to  $\boldsymbol{\sigma}_{\mathcal{C} \setminus T}^{(1)}$. On the other hand for $\boldsymbol{\tau}^{(2)}$, variables covered by $T$ now have values of $\boldsymbol{\sigma}_{T}^{(1)}$, while variables covered by $ \mathcal{C} \setminus T$ keep their values, equal to $\boldsymbol{\sigma}_{\mathcal{C} \setminus T}^{(2)}$.

Considering now the combination of Houdayer clusters $T^c$, we can similarly work out the resulting pair. The first element of the resulting pair is such that variables covered by $\mathcal{C} \setminus T$ are now equal to $\boldsymbol{\sigma}_{\mathcal{C} \setminus T}^{(2)}$ while values covered by $T$ keep their values $\boldsymbol{\sigma}_{T}^{(1)}$, which is overall identical to $\boldsymbol{\tau}^{(2)}$ computed before. Similarly, we get that the second element of the resulting pair is $\boldsymbol{\tau}^{(1)}$, which concludes the proof.
\end{proof}

\section{Full Derivation for \cref{example:1d_graph}}\label{appendix:1d_graph}

In one dimension, site percolation with occupation probability $p\in [0, 1]$ can be solved analytically when assuming $|V| \to +\infty$. In this case, one has that~\cite[Section 2.2]{intro_percolation_theory}:
\begin{itemize}
    \item The percolation threshold is $p_c=1$.
    \item For an arbitrary site, writing $S(p)$ for the random variable denoting the size of the cluster containing it, we have $\mathbb{P}(S(p) = s) = s(1-p)^2p^{s-1}$.
    \item The mean cluster size is $\mathbb{E}[S(p)] =\frac{1+p}{1-p}$.
\end{itemize}
Moreover, the cluster size can be shown to be sharply concentrated at its mean value since for $\varepsilon\geq \mathbb{E}[S(p)]$ we have
\begin{align}
    \mathbb{P}(|S(p)-\mathbb{E}[S(p)]|\geq \varepsilon) &= \mathbb{P}(S(p)-\mathbb{E}[S(p)]\geq \varepsilon)\nonumber \\
    &=\mathbb{P}(S(p)\geq \mathbb{E}[S(p)]+\varepsilon)\nonumber \\
    &=1- \sum_{s=0}^{\varepsilon+\mathbb{E}[S(p)]-1} s(1-p)^2p^{s-1}\nonumber\\
    &= p^{\frac{2p}{1-p}+\varepsilon}\left(1+2p+(1-p)\varepsilon\right)\label{equation:cluster_size_concentration},
\end{align}
which follows an exponential decrease as a function of $\varepsilon$.

Recall from \cref{appendix:percolation_theory} that the formation of Houdayer clusters can be seen as a site percolation process with occupation probability $p=1/2$ when the distribution $\mu$ over pairs is uniform. Consequently in this case, \cref{equation:cluster_size_concentration} yields 
\begin{equation*}
    \mathbb{P}(|S(1/2)-\mathbb{E}[S(1/2)]|\geq \varepsilon) = \frac{4+\varepsilon}{2^{3+\varepsilon}},
\end{equation*}
meaning that all clusters have size $\mathbb{E}[S(1/2)]=3$ with high probability. With this, the number of variables covered by a certain combination $T\subseteq \C$ of Houdayer clusters is given by $|T|\cdot\mathbb{E}[S(1/2)]$  with high probability so that the expected number of flips reads
\begin{align*}
    \eta_{\mu} &= \sum_{T\subseteq \C} \nu_{\C}(T)\cdot\min\left(3|T|, 3(|\C|-|T|)\right)\\
    &= 3\sum_{i=0}^{|\C|}\left(\mathop{\sum_{T\subseteq \C:}}_{|T|=i}\nu_{\C}(T)\right)\cdot\min\left(i, |\C|-i\right).
\end{align*}

To recover an equation which depends on $V$ rather than $\C$, we use the fact that the number of variables with local overlap $-1$ is equal to $|V|p$ in expectation, so that $|V|p = \sum_{C\in\C}|C| = \sum_{C\in\C}\mathbb{E}[S(p)]$, or equivalently $|\C| = \frac{|V|p}{\mathbb{E}[S(p)]}$. In particular:
\begin{itemize}
    \item For the distribution $\nu_{\mathcal{C}}(T) = \mathbbm{1}\{T\in \mathcal{T}_1\} \cdot \frac{\sum_{C\in T} |C|}{\sum_{T\in \mathcal{T}_1}\sum_{C\in T} |C|}$ (original Houdayer move), we have
    \begin{equation*}
        \eta_{\mu} = 3\left(\mathop{\sum_{T\subseteq\C:}}_{|T|=1}\nu_{\C}(T)\right)\cdot\min\left(1, |\C|-1\right) = 3.
    \end{equation*}
    \item For the distribution $\nu_{\mathcal{C}}(T) = \frac{1}{2^{|\C|}}$ (extended Houdayer move), we have
    \begin{equation*}
        \eta_{\mu} = \frac{3}{2^{|\C|}}\sum_{i=0}^{|\C|}\binom{|\C|}{i}\cdot\min\left(i, |\C|-i\right) = \Theta(|V|).
    \end{equation*}
    Indeed, one can upper-bound this expression in terms of the number of nodes using
    \begin{align*}
        \frac{3}{2^{|\C|}}\sum_{i=0}^{|\C|}\binom{|\C|}{i}\cdot\min\left(i, |\C|-i\right) &\leq \frac{3}{2^{|\C|}}\sum_{i=0}^{|\C|}i\binom{|\C|}{i} \\
        &= \frac{3}{2^{|\C|}}\cdot|\C|\cdot 2^{|\C|-1}\\
        &= \frac32 |\C|\\
        &= \frac{|V|}{4}.
    \end{align*}
    To find a lower-bound, we use the fact that 
    \begin{align}
        \sum_{i=0}^{|\C|}\binom{|\C|}{i}\Big||\C|-2i\Big| &= \sum_{i=0}^{|\C|}\sqrt{\binom{|\C|}{i}(|\C|-2i)^2}\sqrt{\binom{|\C|}{i}}\nonumber\\
        &\leq \sqrt{\sum_{i=0}^{|\C|} \binom{|\C|}{i}(|\C|-2i)^2}\cdot \sqrt{\sum_{i=0}^{|\C|} \binom{|\C|}{i}}\label{equation:cauchy_schwartz}\\
        &=\sqrt{|\C|\cdot2^{|\C|}} \cdot \sqrt{2^{|\C|}}\label{equation:binomial_i_squared}\\
        &= 2^{|\C|}\sqrt{|\C|}\nonumber
    \end{align}
    where \cref{equation:cauchy_schwartz} follows from the Cauchy-Schwartz inequality and \cref{equation:binomial_i_squared} is established by repeatedly using the identity $i\binom{|\C|}{i}=|\C|\binom{|\C|-1}{i-1}$. Moreover since $\min(a, b) = \frac12(a+b -|a-b|)$ for $a, b\in \mathbb{R}$, we have
    \begin{align*}
        \frac{3}{2^{|\C|}}\sum_{i=0}^{|\C|}&\binom{|\C|}{i}\cdot\min\left(i, |\C|-i\right) \\
        &= \frac{3}{2^{|\C|}}\sum_{i=0}^{|\C|}\binom{|\C|}{i}\frac12\left(|\C|-\Big||\C|-2i\Big|\right)\\
        &\geq \frac{3}{2^{|\C|+1}}\left(\left(\sum_{i=0}^{|\C|}\binom{|\C|}{i}|\C|\right)-2^{|\C|}\sqrt{|\C|}\right)\\
        &= \frac{3}{2^{|\C|+1}}\left(2^{|\C|}|\C|-2^{|\C|}\sqrt{|\C|}\right)\\
        &= \frac{3|\C|}{2}\left(1-\sqrt{\frac{1}{|\C|}}\right)\\
        &=\frac{|V|}{4}\left(1-\sqrt{\frac{6}{|V|}}\right).
    \end{align*}
\end{itemize}

\bibliography{main}

\end{document}